\documentclass[journal, twocolumn, 10pt]{IEEEtran}
\usepackage{graphicx}
\usepackage{algorithm}
\usepackage{algorithmicx}
\usepackage{algpseudocode}
\usepackage{cite}
\usepackage{color}
\usepackage{amsmath,amssymb,amsfonts,amsthm}
\usepackage{bm}
\usepackage{setspace}
\usepackage{multirow}
\usepackage{mathtools}
\usepackage{epstopdf}

\newtheorem{thm} {Theorem}

\newtheorem{lemma} {Lemma}

\newcommand{\lb}{\left(}
\newcommand{\rb}{\right)}
\newcommand{\cin}{c_{\text{in}}}
\newcommand{\cout}{c_{\text{out}}}
\newcommand{\pin}{p_{\text{in}}}
\newcommand{\pout}{p_{\text{out}}}
\newcommand{\nin}{n_{\text{in}}}
\newcommand{\nout}{n_{\text{out}}}
\newcommand{\Ain}{\mathbf{A}_{\text{in}}}
\newcommand{\Aout}{\mathbf{A}_{\text{out}}}
\newcommand{\C}{\mathbf{C}}
\newcommand{\bCbar}{\mathbf{\overline{C}}}
\newcommand{\onein}{\mathbf{1}_{\text{in}}}
\newcommand{\oneout}{\mathbf{1}_{\text{out}}}
\newcommand{\Din}{\mathbf{D}_{\text{in}}}
\newcommand{\Dout}{\mathbf{D}_{\text{out}}}
\newcommand{\Lin}{\mathbf{L}_{\text{in}}}
\newcommand{\Lout}{\mathbf{L}_{\text{out}}}
\newcommand{\yin}{\mathbf{y}_{\text{in}}}
\newcommand{\yout}{\mathbf{y}_{\text{out}}}
\newcommand{\xin}{\mathbf{x}_{\text{in}}}
\newcommand{\xout}{\mathbf{x}_{\text{out}}}
\newcommand{\bx}{\mathbf{x}}
\newcommand{\by}{\mathbf{y}}
\newcommand{\ra}{\rightarrow}
\begin{document}
\title{Deep Community Detection}

\author{Pin-Yu~Chen and Alfred O. Hero III,~\emph{Fellow},~\emph{IEEE}
\thanks{P.-Y. Chen and A. O. Hero are with the Department of Electrical Engineering and Computer Science, University of Michigan, Ann Arbor, MI 48109, USA. Email : pinyu@umich.edu and hero@umich.edu.}
\thanks{
Some results in this paper were presented in part at IEEE GLOBALSIP 2013 \cite{CPY13GlobalSIP} and IEEE ICASSP 2014 \cite{CPY14ICASSP}.}
\thanks{This work has been partially supported by the Army Research Office (ARO), grant number W911NF-12-1-0443.}
}

\maketitle
\thispagestyle{empty}
\begin{abstract}
A deep community in a graph is a connected component that can only be seen after removal of nodes or edges from the rest of the graph. This paper formulates the problem of detecting deep communities as multi-stage node removal that maximizes a new centrality measure, called the local Fiedler vector centrality (LFVC), at each stage.
The LFVC is associated with the sensitivity of algebraic connectivity to node or edge removals.
We prove that a greedy node/edge removal strategy,  based on successive maximization of LFVC, has bounded performance loss relative to the optimal,  but intractable, combinatorial batch removal strategy. Under a stochastic block model framework, we show that the greedy LFVC strategy can  extract deep communities with probability one as the number of observations becomes large. We apply the greedy LFVC strategy to real-world social network datasets.
Compared with conventional community detection methods  we demonstrate improved ability to identify
important communities and key members in the network.
\end{abstract}


\section{Introduction}
\label{sec_Intro}

In social, biological and technological network analysis \cite{Girvan02,Bertrand13Mag,Shuman13}, community detection aims to extract tightly connected subgraphs in the networks.  This problem has attracted a great deal of interest in network science \cite{Fortunato09,Fortunato10}.
Community detection is often cast as graph partitioning. Many graph partitioning methods exist in the literature, including graph cuts \cite{Shi00,Luxburg07}, probabilistic models \cite{Snijders97,Karrer11}, and node/edge pruning strategies based on different criteria \cite{Girvan02,Newman06PNAS,Coscia11,Wen11}.

Many community detection methods are based on detecting nodes or edges with high centrality. Node and edge centralities are quantitative measures that are used to evaluate the level of importance and/or influence of a node or an edge in the network. Centralities can be based on combinatorial measures such as shortest paths or graph diffusion distances between every node pair \cite{Wasserman94,Newman10NetworkIntro}. Centrality measures can also be based on spectral properties of the adjacency and graph Laplacian matrices associated with the graph \cite{Newman10NetworkIntro}.
Many of these measures require global topological information and therefore may not be computationally feasible for very large networks.

Nonparametric community detection methods, such as the the edge betweenness method \cite{Girvan02} and the modularity method \cite{Newman06PNAS}, can be viewed as edge removal strategies that aim to maximize a centrality measure, e.g., the modularity or  betweenness measures.  It is worth noting that these methods
presume that each node in the graph is affiliated with a community.
However, in some community detection applications it often occurs that the graphs contain spurious edges connecting to irrelevant ``noisy'' nodes that are not members of any single community.
In such cases, noisy nodes and edges mask the true communities in the graph. Detection of these masked communities is a difficult problem that we call ``deep community detection''. The formal definition of a deep community is given in Sec. \ref{Sec_deep}. Due to the presence of noisy nodes and spurious edges \cite{Fortunato07,Balakrishnan11},
deep communities elude detection when conventional community detection methods methods are applied.

In this paper, a new partitioning strategy is applied to detect deep communities. This strategy uses a new local measure of centrality that is specifically designed to unmask communities in the presence of spurious edges.
The new partitioning strategy is based on a novel spectral measure \cite{Chung97SpectralGraph} of centrality called local Fiedler vector centrality (LFVC).
LFVC is associated with the sensitivity of algebraic connectivity \cite{Fiedler73} when a subset of nodes or edges are removed from a graph \cite{CPY14ICASSP,CPY14ComMag}.
We show that LFVC relates to a monotonic submodular set function which ensures that greedy node or edge removals based on LFVC are nearly as effective as the optimal combinatorial batch removal strategy.

Our approach utilizes LFVC to iteratively remove nodes in the graph to reveal deep communities.
A removed node that connects multiple deep communities is assigned mixed membership: it is shared among these communities.
Under a ``signal plus noise'' stochastic block model framework \cite{Holland83,Snijders97,Karrer11}, we
use random matrix theory to show that the greedy LVFC strategy can asymptotically identity the deep communities with probability one.
 As compared with the modularity method \cite{Newman06PNAS} and the L1 norm subgraph detection method \cite{MillerICASSP10,Miller10}, we show that the proposed greedy LFVC approach has superior deep community detection performance.
We illustrate the proposed deep community detection method on several real-world social networks.
When our proposed greedy LFVC approach is applied to the network scientist coauthorship dataset \cite{Newman06community}, it reveals deep communities that are not identified by conventional community detection methods.
When applied to social media, the Last.fm online music dataset, we show that
LFVC has the best performance in detecting users with similar interest in artists.

The rest of this paper is organized as follows. Sec. \ref{Sec_spectral} summarizes commonly used centrality measures, the definitions of community, and relevant spectral graph theory.  Sec. \ref{Sec_deep} gives a definition of deep communities. The proposed local Fiedler vector centrality (LFVC) is defined in Sec. \ref{Sec_dist_centrality}.
In Sec. \ref{Sec_block}, we introduce the signal plus noise stochastic block model for a deep community and
establish asymptotic community detection performance of the greedy LFVC strategy.
We apply the greedy LFVC strategy to  real-world social network datasets in Sec. \ref{Sec_performance}. Finally, Sec. \ref{sec_con} concludes the paper.
Throughout the paper we use uppercase letters in boldface (e.g., $\mathbf{A}$) to represent matrices, lowercase letters in boldface (e.g., $\mathbf{a}$) to represent vectors, and uppercase letters in calligraphic face (e.g., $\mathcal{A}$) to represent sets. Subscripts on matrices and vectors indicate elements (e.g., $\mathbf{A}_{ij}$ is the element of $i$-th row and $j$-th column of matrix $\mathbf{A}$, and $a_i$ is the $i$-th element of vector $\mathbf{a}$). $(\cdot)^T$ denotes matrix and vector transpose.

\section{Centralities, Communities, and Spectral Graph Theory}
\label{Sec_spectral}

\subsection{The graph Laplacian matrix and algebraic connectivity}
Consider an undirected and unweighted graph $G=(\mathcal{V},\mathcal{E})$ without self loops or multiple edges. We denote by $\mathcal{V}$ the node set, with $|\mathcal{V}|=n$, and by $\mathcal{E}$ the edge set, with $|\mathcal{E}|=m$. The connectivity structure of $G$ is characterized by an $n$-by-$n$ binary symmetric adjacency matrix $\mathbf{A}$, where $\mathbf{A}_{ij}=1$ if $(i,j) \in \mathcal{E}$, otherwise $\mathbf{A}_{ij}=0$. Let $d_i=\sum_{j=1}^n \mathbf{A}_{ij}$ denote the degree of node $i$. The degree matrix $\mathbf{D}=\textnormal{diag}(\mathbf{d})$ is a diagonal matrix with the degree vector $\mathbf{d}=[d_1,d_2,\ldots d_n]$ on its diagonal.
The graph Laplacian matrix of $G$ is defined as $\mathbf{L}=\mathbf{D}-\mathbf{A}$. Let $\lambda_i(\mathbf{L})$ denote the $i$-th smallest eigenvalue of $\mathbf{L}$ and let $\textbf{1}=[1,\ldots,1]^T$ denote the vector of ones. We have the representation \cite{Merris94,Chung97SpectralGraph}
\begin{align}
\label{eqn_alge_quadratic}
\mathbf{x}^T \mathbf{L} \mathbf{x}=\frac{1}{2}\sum_{i \in \mathcal{V}} \sum_{j \in \mathcal{V}} \mathbf{A}_{ij} (x_i-x_j)^2,
\end{align}
which is nonnegative,
and $\mathbf{L} \textbf{1}=(\mathbf{D}-\mathbf{A})\textbf{1}=\textbf{0}$, the vector of all zeros. Therefore
$\lambda_1(\mathbf{L})=0$ and $\mathbf{L}$ is a positive semidefinite (PSD) matrix.

The algebraic connectivity of $G$ is defined as the second smallest eigenvalue of $\mathbf{L}$, i.e., $\lambda_2(\mathbf{L})$. $G$ is connected if and only if $\lambda_2(\mathbf{L})>0$. Moreover, it is a well-known property \cite{Fiedler73} that for any non-complete graph,
\begin{align}
\label{eqn_connectivity}
\lambda_2(\mathbf{L}) \leq \textnormal{node connectivity} \leq \textnormal{edge connectivity},
\end{align}
where node/edge connectivity is the least number of node/edge removals that disconnects the graph. (\ref{eqn_connectivity}) is the main motivation for our proposed node/edge pruning approach. A graph with larger algebraic connectivity is more resilient to node and edge removals. In addition, let $d_{\min}$ be the minimum degree of $G$, it is also well-known \cite{Abreu07,Chung97SpectralGraph} that $\lambda_2(\mathbf{L}) \leq 1$ if and only if $d_{\min}=1$. That is, a graph with a leaf node (i.e., a node with a single edge) cannot have algebraic connectivity larger than 1. For any connected graph, we can represent the algebraic connectivity as
\begin{align}
\label{eqn_alge}
\lambda_2(\mathbf{L})=\min_{\|\mathbf{x}\|_2=1,~\mathbf{x} \perp \textbf{1}}{\mathbf{x}^T \mathbf{L} \mathbf{x}}
\end{align}
by the Courant-Fischer theorem \cite{HornMatrixAnalysis} and the fact that the constant vector is the eigenvector associated with $\lambda_1(\mathbf{L})=0$.

\subsection{Some examples of centralities}
\label{subsec_centrality}
Centrality measures can be classified into two categories, \emph{global} and \emph{local} measures. Global centrality measures require complete topological information for their computation, whereas local centrality measures only require local topological information from neighboring nodes.
Some examples of node centralities are:
\begin{itemize}
  \item \textbf{Betweenness} \cite{Freeman77}: betweenness measures the fraction of shortest paths passing through a node relative to the total number of shortest paths in the network. Specifically, betweenness
      is a global measure defined as $\text{betweenness}(i)=\sum_{k \neq i} \sum_{j \neq i, j > k} \frac{\phi_{kj}(i)}{\phi_{kj}}$, where
$\phi_{kj}$ is the total number of shortest paths from $k$ to $j$ and $\phi_{kj}(i)$ is the number of such shortest paths passing through $i$. A similar notion is used to define the edge betweenness centrality \cite{Girvan02}.

  \item \textbf{Closeness} \cite{Sabidussi66Closeness}: closeness is a global measure of geodesic distance of a node to all other nodes. A node is said to have high closeness if the sum of its shortest path distances to other nodes is small.
  Let $\rho(i,j)$ denote the shortest path distance between node $i$ and node $j$ in a connected graph. Then we define $\text{closeness}(i)=1/{\sum_{j \in \mathcal{V}, j\neq i} \rho(i,j)}$.

 \item \textbf{Eigenvector centrality} (eigen centrality) \cite{Newman10NetworkIntro}:  eigenvector centrality is the $i$-th entry of the eigenvector associated with the largest eigenvalue of the adjacency matrix $\mathbf{A}$. It
  is defined as \text{eigen}$(i)=\lambda_{\max}^{-1} \sum_{j\in \mathcal V} \mathbf A_{ij}\xi_{j}$, where $\lambda_{\max}$ is the largest eigenvalue of $\mathbf{A}$ and $\mathbf{\xi}$ is the eigenvector associated with $\lambda_{\max}$.
    It is a global measure since eigenvalue decomposition on $\mathbf{A}$ requires global knowledge of the graph topology.
  \item \textbf{Degree} ($d_i$): degree is the simplest local node centrality measure which accounts for the number of neighboring nodes.
  \item \textbf{Ego centrality} \cite{Everett05Ego}: consider the $(d_i+1)$-by-$(d_i+1)$ local adjacency matrix of node $i$, denoted by $\mathbf{A}(i)$, and let $\mathbf{I}$ be an identity matrix.
  Ego centrality can be viewed as a local version of betweenness that computes the shortest paths between its neighboring nodes.
      Since $[\mathbf{A}^2(i)]_{kj}$ is the number of two-hop walks between $k$ and $j$, and $\left[\mathbf{A}^2(i) \circ \left(\mathbf{I}-\mathbf{A}(i)\right) \right]_{kj}$ is the total number of two-hop shortest paths between $k$ and $j$ for all $k \neq j$, where $\circ$ denotes entrywise matrix product. Ego centrality is defined as
      $\text{ego}(i)=\sum_{k} \sum_{j>k} 1/{\left[\mathbf{A}^2(i) \circ \left(\mathbf{I}-\mathbf{A}(i)\right) \right]_{kj}}$.
\end{itemize}
These centrality measures are used for comparison with the proposed centrality measure (LFVC) for deep community detection in Sec. \ref{Sec_performance}.

\subsection{Community and modularity}
\label{subsec_community}
Many possible definitions of communities exist in the literature \cite{Wasserman94,Fortunato10,Coscia11}. One widely adopted definition is based on the relations between the number of internal and external connections of a subgraph $S \subset G$  \cite{Radicchi02}. For a subgraph $S \subset G$ with node set $\mathcal{V}_S$, let $d_i^{\textnormal{int}}(S)=\sum_{j \in S} \mathbf{A}_{ij}$ denote the number of internal edges of node $i$ in $S$ and $d_i^{\textnormal{ext}}(S)=\sum_{j \in \mathcal{V}/\mathcal{V}_S} \mathbf{A}_{ij}$ denote the number of external edges of node $i$ outside $S$. $S$ is said to be a community in the strong sense if $d_i^{\textnormal{int}}(S)>d_i^{\textnormal{ext}}(S)$ for all $i \in S$, and $S$ is said to be a community in the weak sense if
$\sum_{i\in S} d_i^{\textnormal{int}}(S)>\sum_{i \in S}d_i^{\textnormal{ext}}(S)$.

Newman \cite{Newman06PNAS} defines a community by comparing the internal and external connections of a subgraph with that of a random graph having the same degree pattern {(i.e., a random graph where each node has exactly the same degree as the original graph), and he proposes a quantity called modularity to construct a graph partitioning of $G$ into communities.
First consider partitioning a graph into two communities. Recalling that $m=|\mathcal{E}|$ is the number of edges in the graph, define $\mathbf{B}_{ij}=\mathbf{A}_{ij}-\frac{d_id_j}{2m}$. $\mathbf{B}_{ij}$ can be interpreted as the number of excessive edges between $i$ and $j$ since $\mathbf{B}_{ij}$ is the difference of actual edges minus the expected edges of the degree-equivalent random graph. Let $\mathbf{s}$ be the membership vector such that $s_i=1$ if $i$ is in community $1$ and $s_i=-1$ if $i$ is in community $2$. Modularity $Q$ is proportional to the total number of excess edges in each community, i.e.,
\begin{align}
\label{eqn_modularity}
Q=\frac{1}{4m} \sum_{i \in \mathcal{V}} \sum_{j \in \mathcal{V}} \mathbf{B}_{ij} (s_i s_j+1) = \frac{1}{4m} \mathbf{s}^T \mathbf{B} \mathbf{s}
\end{align}
since $\sum_{i \in \mathcal{V}} \sum_{j \in \mathcal{V}} \mathbf{B}_{ij}=0$. Maximizing this quadratic form yields a partition of $G$ into two communities \cite{Newman06community}.
The associated membership vector $\mathbf{s}$ can be obtained by computing the largest eigenvector $\mathbf {b}_{\max}$ of $\mathbf{B}$ and extracting its polarity, i.e., $\mathbf{s}={\mathrm{sign}}(\mathbf {b}_{\max})$ \cite{Newman06community}.

To divide a network into more than two communities, Newman proposes a recursive partitioning approach. It is also verified in \cite{Newman13} that there is no performance difference between the modularity method \cite{Newman06PNAS}, the statistical inference method \cite{Snijders97,Airoldi08}, and the normalized cut method \cite{Luxburg07}.
However, the modularity method may fail to detect small communities even when community structures are apparent \cite{Fortunato07,Good10,Nadakuditi12Detecability}.

In \cite{Wen11}, a node removal strategy based on targeting high degree nodes is proposed to improve the performance of the modularity method. The authors of \cite{Wen11} argue that high-degree nodes incur more noisy connections than low-degree nodes, and it is experimentally demonstrated that removing high-degree nodes can better reveal the community structure.

\subsection{The Fiedler vector}
The Fiedler vector of a graph is the eigenvector associated with the second smallest eigenvalue $\lambda_2(\mathbf{L})$ of the graph Laplacian matrix $\mathbf{L}$ \cite{Fiedler73}. The Fiedler vector has been widely used in graph partitioning, image segmentation and data clustering \cite{Pothen90,Spielman2007,Shi00,Luxburg07,Schaeffer07}. Analogously to modularity partitioning, the Fiedler vector performs community detection by separating the nodes in the graph according to the signs of the corresponding Fiedler vector elements. Similarly, hierarchical community structure can be detected by recursive partitioning with the Fiedler vector.

In this paper, we use the Fiedler vector to define a new centrality measure.
One advantage of using the Fiedler vector over other global centrality measures is that
it can be computed in a distributed manner via local information exchange over the graph \cite{Bertrand13}.

\section{Deep community}
\label{Sec_deep}
A deep community is defined in terms of an additive signal (community) plus noise model. Let $\mathbf{A}_1,\ldots,\mathbf{A}_g$ denote the $n \times n$ mutually orthogonal binary adjacency matrices associated with $g$ non-singleton connected components in a noiseless graph $G_0$ over $n$ nodes. Assume the nodes have been permuted so that $\mathbf{A}_1,\ldots,\mathbf{A}_g$ are block diagonal with non-overlapping block indices $\mathcal{I}_1,\ldots,\mathcal{I}_g$. The observed graph $G$ is a noise corrupted version of $G_0$ where random edges have been inserted between the connected components of $G_0$. More specifically, let $\mathbf{A}_{\text{nse}}$ be a random adjacency matrix with the property that $\mathbf{A}_{\text{nse}}(i,j)=0$, $i,j \in \mathcal{I}_k$, for $k=1,\ldots,g$ and where the rest of the elements of $\mathbf{A}_{\text{nse}}$ are Bernoulli i.i.d random variables. Then the adjacency matrix $\mathbf{A}$ of $G$ satisfies the signal plus noise model
\begin{align}
\label{eqn_Adjacency_Boolean}
\mathbf{\mathbf{A}}=\sum_{k=1}^g  \mathbf{A}_k  + \mathbf{A}_{\text{nse}}.
\end{align}

The deep community detection problem is to recover connected components $\mathbf{A}_1,\ldots,\mathbf{A}_g$ from the noise corrupted observations $\mathbf{A}$. The $\mathbf{A}_k$'s are called deep communities in the sense that they are embedded
in a graph with random interconnections between connected components. The performance analysis of deep community detection on networks generated by a specified stochastic block model \cite{Holland83}
is discussed in Sec. \ref{Sec_block}. An illustrative visual example of deep community detection is shown in the longer arXiv version of this paper\footnote{available at http://arxiv.org/abs/1407.6071}.
Deep community detection is equivalent to the planted clique problem \cite{Alon98} in the special case that $g=1$ and the non-zero block of $\mathbf A_1$ corresponds to a complete graph, i.e., all off-diagonal elements of this block are equal to one.   
Models similar to (\ref{eqn_Adjacency_Boolean}) have also been used for hypothesis testing on the existence of dense subgraphs embedded in random graphs \cite{MillerICASSP10,Miller10}. The null hypothesis is the noise only model (i.e., $\mathbf{A}_k=\mathbf{0}~\forall k$). The alternative hypothesis is the signal plus noise model (\ref{eqn_Adjacency_Boolean}) with $\mathbf{A}_k \neq \mathbf{0}$.  The authors in \cite{MillerICASSP10,Miller10} propose to use the L1 norms of the eigenvectors of the modularity matrix $\mathbf{B}$ as test statistics. This statistic is compared with our proposed local Fiedler vector deep community detection method in Sec. \ref{Sec_block}.

We propose an iterative denoising algorithm for recovering deep communities that is based on either node or edge removals. The proposed algorithm uses a spectral centrality measure, defined in Sec. \ref{Sec_dist_centrality}, to determine the nodes/edges to be pruned from the observed graph with adjacency matrix $\mathbf{A}$. 

%

Let $\widetilde{\mathbf{L}}$ be the resulting $n \times n$ graph Laplacian matrix after removing a subset of nodes or edges from the graph.
The following theorem provides an upper bound on the number of deep communities in the remaining graph $\widetilde{G}$.
\begin{thm}
\label{Thm_number_community}
For any node removal set $\mathcal{R}$ of $G$ with $|\mathcal{R}|=q$, let $r$ be the rank of the resulting graph Laplacian matrix $\widetilde{\mathbf{L}}$ and let $\|\widetilde{\mathbf{L}}\|_\ast=\sum_i {\lambda_i(\widetilde{\mathbf{L}})}$ denote its nuclear norm.
The number $\epsilon$ of remaining non-singleton connected components in $\widetilde{G}$ has the upper bound
\begin{align}
\label{eqn_UB_of_deep_community}
\epsilon &\leq n-q-r  \nonumber \\
&\leq n-q-\frac{\|\widetilde{\mathbf{L}}\|_\ast}{\lambda_n(\widetilde{\mathbf{L}})} \nonumber \\
&=n-q-\frac{2\widetilde{m}}{\lambda_n(\widetilde{\mathbf{L}})},
\end{align}
where $\widetilde{m}$ is the number of edges in $\widetilde{G}$.
The first inequality in (\ref{eqn_UB_of_deep_community}) becomes an equality if all connected components in $\widetilde{G}$ are non-singletons.
The second inequality in (\ref{eqn_UB_of_deep_community}) becomes an equality if all non-singleton connected components are complete subgraphs of the same size. Similarly, for any edge removal set of $G$, let $r$ be the rank of the resulting graph Laplacian matrix $\widetilde{\mathbf{L}}$. The number $\epsilon$ of remaining non-singleton connected components in $\widetilde{G}$ has the upper bound
$\epsilon \leq n-r
\leq n-\frac{\|\widetilde{\mathbf{L}}\|_\ast}{\lambda_n(\widetilde{\mathbf{L}})}
=n-\frac{2\widetilde{m}}{\lambda_n(\widetilde{\mathbf{L}})}.$
\end{thm}
\begin{proof}
The proof can be found in Appendix \ref{appen_num_community_proof}.
\end{proof}

The upper bound in Theorem \ref{Thm_number_community} can be further relaxed by applying the inequality $\lambda_n(\widetilde{\mathbf{L}}) \leq 2 \widetilde{d}_{\max}$ \cite{Chung97SpectralGraph}, where $\widetilde{d}_{\max}$ is the maximum degree of $\widetilde{G}$. Other bounds on $\lambda_n(\widetilde{\mathbf{L}})$ can be found in \cite{Guo05}.

The next theorem shows that the largest non-singleton connected component size can be represented as a matrix one norm of a matrix whose column vectors are orthogonal and sparsest among all binary vectors that form a basis of the null space of $\widetilde{\mathbf{L}}$.

\begin{thm}
\label{Thm_largest_component_size}
Define the sparsity of a vector to be the number of zero entries in the vector.
Let $\textnormal{null}(\widetilde{\mathbf{L}})$ denote the null space of $\widetilde{\mathbf{L}}$ and let $\mathbf{X}$ denote the matrix whose columns are orthogonal and they form the sparsest basis of \textnormal{null}$(\widetilde{\mathbf{L}})$ among binary vectors.
Let $\psi(\widetilde{G})$ be the largest non-singleton connected component size of $\widetilde{G}$.
Then $\psi(\widetilde{G})=\| \mathbf{X}\|_1=\max_i \|\mathbf{x}_i\|_1$, where $\mathbf{x}_i$ is the $i$-th column vector of binary matrix $\mathbf{X}$.
\end{thm}
\begin{proof}
The proof can be found in Appendix \ref{appen_largest_component_proof}.
\end{proof}

Theorems \ref{Thm_number_community} and \ref{Thm_largest_component_size} are key results that motivate and theoretically justify the proposed local Fiedler vector centrality measure introduced below. Theorem \ref{Thm_number_community} establishes that the number of deep communities is closely related to the number of edge/node removals that are required to reveal them. Theorem \ref{Thm_largest_component_size} establishes that $L_1$ norm of the sparsest basis for the null space of the graph Laplacian matrix can be used to estimate the size of the largest deep community in the network.

\section{The Proposed Node and Edge Centrality: Local Fiedler Vector Centrality (LFVC)}
\label{Sec_dist_centrality}
The proposed deep community detection algorithm (Algorithm \ref{algo_LFVC_detection}) is based on removal of nodes or edges according to how the removals affect a measure of algebraic connectivity. This measure, called the local Fiedler vector centrality (LFVC), is computed from the graph Laplacian matrix.  
In particular, the LFVC is motivated by the fact  that node/edge removals result in low rank perturbations to the graph Laplacian matrix when $n \gg d_{\max}$, where $d_{\max}$ is the maximum degree. The node and edge LFVC are then defined to correspond to an upper bound on algebraic connectivity.

\subsection{Edge-LFVC}
\label{subsec_rankone_perturbation}
Considering the graph $\widetilde{G}(i,j)=(\mathcal{V},\mathcal{E} \cup (i,j))$ by adding an edge $(i,j)\notin \mathcal{E}$ to $G$,
we have $\widetilde{\mathbf{L}}=\mathbf{L}+\Delta \mathbf{L}$ and $\Delta \mathbf{L}=\Delta \mathbf{D}-\Delta \mathbf{A}$, where
$\Delta \mathbf{D}$ and $\Delta \mathbf{A}$ are the augmented degree and adjacency matrices, respectively.
Denote the resulting graph Laplacian matrix by $\widetilde{\mathbf{L}}(i,j)$.
Let $\mathbf{e}_i$ be a zero vector except that its $i$-th element is equal to $1$. Then
\begin{align}
&\Delta \mathbf{D}=\textnormal{diag}(\mathbf{e}_i)+\textnormal{diag}(\mathbf{e}_j)=\mathbf{e}_i \mathbf{e}_i^T + \mathbf{e}_j \mathbf{e}_j^T; \\
&\Delta \mathbf{A}=\mathbf{e}_i \mathbf{e}_j^T + \mathbf{e}_j \mathbf{e}_i^T,
\end{align}
and therefore
\begin{align}
\widetilde{\mathbf{L}}(i,j)=\mathbf{L}+(\mathbf{e}_i-\mathbf{e}_j)(\mathbf{e}_i-\mathbf{e}_j)^T.
\end{align}
Thus, the resulting graph Laplacian matrix $\widetilde{\mathbf{L}}(i,j)$ after adding an edge $(i,j)$ to $G$ is the original graph Laplacian matrix $\mathbf{L}$ perturbed by a rank one matrix $(\mathbf{e}_i-\mathbf{e}_j)(\mathbf{e}_i-\mathbf{e}_j)^T$.
Similarly, when an edge  $(i,j) \in \mathcal{E}$ is removed from $G$, we have $\widetilde{\mathbf{L}}(i,j)=\mathbf{L}-(\mathbf{e}_i-\mathbf{e}_j)(\mathbf{e}_i-\mathbf{e}_j)^T$.

Consider removing an edge $(i,j) \in \mathcal{E}$ from $G$ resulting in $\widetilde{\mathbf{L}}(i,j)$ above.
Let $\mathbf{y}$ denote the Fiedler vector of $\mathbf{L}$, computing $\mathbf{y}^T \widetilde{\mathbf{L}}(i,j) \mathbf{y}$ gives an upper bound on $\lambda_2(\widetilde{\mathbf{L}}(i,j))$ as
\begin{align}
\label{eqn_UB_alge_link}
\lambda_2(\widetilde{\mathbf{L}}(i,j)) &\leq \mathbf{y}^T \widetilde{\mathbf{L}}(i,j) \mathbf{y} \nonumber \\
&= \mathbf{y}^T (\mathbf{L}-(\mathbf{e}_i-\mathbf{e}_j)(\mathbf{e}_i-\mathbf{e}_j)^T) \mathbf{y} \nonumber \\
& =\lambda_2(\mathbf{L})-(y_i-y_j)^2
\end{align}
following the definition of $\lambda_2(\mathbf{L})=\min_{\|\mathbf{x}\|_2=1, \mathbf{x} \perp \textbf{1}} \mathbf{x}^T \mathbf{L} \mathbf{x}$ in (\ref{eqn_alge}).
It is worth mentioning that for any connected graph $G$ there exists at least one edge removal such that the inequality $\lambda_2(\widetilde{\mathbf{L}}(i,j)) < \lambda_2(\mathbf{L}) $ holds, otherwise $y_i=y_j$ for all
$i,j \in \mathcal{V}$ and this violates the constraints that $\|\mathbf{y}\|_2=1$ and $\sum_{i=1}^n y_i=0$. Consequently, there exists at least one edge removal that leads to a decrease in algebraic connectivity.

Similarly, when we remove a subset of edges $\mathcal{E}_\mathcal{R} \subset \mathcal{E}$ from $G$, where $|\mathcal{E}_\mathcal{R}|=h$. We obtain an upper bound
\begin{align}
\label{eqn_UB_alge_multiple_link}
\lambda_2(\widetilde{\mathbf{L}}(\mathcal{E}_\mathcal{R})) \leq \lambda_2(\mathbf{L})-\sum_{(i,j) \in \mathcal{E}_\mathcal{R}} (y_i-y_j)^2.
\end{align}
Correspondingly, we define the local Fiedler vector edge centrality as
\begin{align}
\label{edge_LFVC}
\textnormal{edge-LFVC}(i,j)=(y_i-y_j)^2.
\end{align}
Edge-LFVC is a measure of centrality as it associates the sensitivity of algebraic connectivity to edge removal as described in (\ref{eqn_UB_alge_multiple_link}).
The top $h$ edge removals which lead to the largest decrease on the right hand side of (\ref{eqn_UB_alge_multiple_link}) are the $h$ edges with the highest edge-LFVC.

\subsection{Node-LFVC}
When a node $i \in \mathcal{V}$ is removed from $G$, all the edges attached to $i$ will also be removed from $G$. Similar to (\ref{eqn_UB_alge_link}), the resulting graph Laplacian matrix $\widetilde{\mathbf{L}}(i)$ can be regarded as a rank $d_i$ matrix perturbation of $\mathbf{L}$.
Since $\mathbf{L}-\widetilde{\mathbf{L}}(i)=\sum_{j\in \mathcal{N}_i} (\mathbf{e}_i-\mathbf{e}_j)(\mathbf{e}_i-\mathbf{e}_j)^T$, where $\mathcal{N}_i$ is the set of neighboring nodes of node $i$, we obtain an upper bound
\begin{align}
\label{eqn_UB_alge_node}
\lambda_2(\widetilde{\mathbf{L}}(i))
&\leq \mathbf{y}^T \widetilde{\mathbf{L}}(i) \mathbf{y} \nonumber \\
&= \mathbf{y}^T (\mathbf{L}+\widetilde{\mathbf{L}}(i)-\mathbf{L}) y \nonumber \\
& =\lambda_2(\mathbf{L})-\sum_{j\in \mathcal{N}_i} (y_i-y_j)^2.
\end{align}
Similar to edge removal, for any connected graph, there exists at least one node removal that leads to a decrease in algebraic connectivity.

If a subset of nodes $\mathcal{R} \subset \mathcal{V}$ are removed from $G$, where $|\mathcal{R}|=q$, then
\begin{align}
\mathbf{L}-\widetilde{\mathbf{L}}(\mathcal{R}) &=\sum_{i \in \mathcal{R}}\sum_{j \in \mathcal{N}_i} (\mathbf{e}_i-\mathbf{e}_j)(\mathbf{e}_i-\mathbf{e}_j)^T \\
&~~~-\frac{1}{2}\sum_{i \in \mathcal{R}}\sum_{j \in \mathcal{R}} \mathbf{A}_{ij}(\mathbf{e}_i-\mathbf{e}_j)(\mathbf{e}_i-\mathbf{e}_j)^T \nonumber,
\end{align}
where the last term accounts for the edges that are attached to the removed nodes at both ends.
Consequently, similar to (\ref{eqn_UB_alge_multiple_link}), we obtain an upper bound for multiple node removals
\begin{align}
\label{eqn_UB_multiple_node}
\lambda_2(\widetilde{\mathbf{L}}(\mathcal{R})) &\leq \lambda_2(\mathbf{L})-\sum_{i \in \mathcal{R}}\sum_{j\in \mathcal{N}_i} (y_i-y_j)^2 \\
&~~~+\frac{1}{2} \sum_{i \in \mathcal{R}}\sum_{j \in \mathcal{R}}\mathbf{A}_{ij} (y_i-y_j)^2. \nonumber
\end{align}
We define the local Fiedler vector node centrality as
\begin{align}
\label{node_LFVC}
\textnormal{node-LFVC}(i)=\sum_{j\in \mathcal{N}_i} (y_i-y_j)^2,
\end{align}
which is the sum of the square terms of the Fiedler vector elementwise differences between node $i$ and its neighboring nodes, and it is also the sum of edge-LFVC of $i's$ neighboring nodes. From (\ref{eqn_UB_alge_node}) and (\ref{eqn_UB_multiple_node}), node-LFVC is associated with the upper bound on the resulting algebraic connectivity for node removal when $|\mathcal{R}|=1$.
A node with higher centrality implies that it plays a more important role in the network connectivity structure.

\subsection{Monotonic submodularity and greedy removals}
Fixing $|\mathcal{R}|=q$, consider the problem of finding the optimal node removal set $\mathcal{R}_{\text{opt}}$ that maximizes the decrease in  the upper bound on algebraic connectivity in (\ref{eqn_UB_multiple_node}). The computational complexity of this batch removal problem is of combinatorial order $\binom{n}{q}$.
Here we show that the greedy LFVC removal procedure, shown in Algorithm \ref{algo_LFVC_detection}, and whose computation is only linear in $n$, has bounded performance loss relative to the combinatorial algorithm in terms of achieving, within a multiplicative constant $(1-1/e)$, an upper bound on algebraic connectivity, where $e$ is Euler's constant.
Let
\begin{align}
\label{eqn_fR}
f(\mathcal{R})=\sum_{i \in \mathcal{R}}\sum_{j\in \mathcal{N}_i} (y_i-y_j)^2-\frac{1}{2} \sum_{i \in \mathcal{R}}\sum_{j \in \mathcal{R}}\mathbf{A}_{ij}
(y_i-y_j)^2
\end{align}
 and recall from (\ref{eqn_UB_multiple_node}) that $\lambda_2(\widetilde{\mathbf{L}}(\mathcal{R})) \leq \lambda_2(\mathbf{L}) - f(\mathcal{R})$. Note that when $|\mathcal{R}|=1$, $f(\mathcal{R})$ reduces to node-LFVC as $\mathbf{A}_{ii}=0$.
The following lemma provides the cornerstone to Theorem \ref{Thm_submodular}.
\begin{lemma}
\label{Lemma_nonnegative}
The function $f(\mathcal{R})$ in (\ref{eqn_fR}) is equal to
\begin{align}
f(\mathcal{R})=\frac{1}{2} \sum_{i \in \mathcal{R}}\sum_{j\in \mathcal{N}_i} (y_i-y_j)^2
+ \frac{1}{2} \sum_{i \in \mathcal{R}}\sum_{j \in \mathcal{V}/\mathcal{R}} \mathbf{A}_{ij} (y_i-y_j)^2. \nonumber
\end{align}
Furthermore, $f(\mathcal{R}) \geq 0$ and $f(\varnothing)=0$, where $\varnothing$ is the empty set.
\end{lemma}
\begin{proof}
The proof can be found in Appendix \ref{appen_nonnegative_proof}.
\end{proof}
The following theorem establishes monotonic submodularity  \cite{krause2012submodular} of $f(\mathcal{R})$. Monotonicity means $f(\mathcal{R})$ is a non-decreasing function: for any subsets $\mathcal{R}_1,\mathcal{R}_2$ of the node set $\mathcal{V}$ satisfying $\mathcal{R}_1 \subset \mathcal{R}_2$ we have $f(\mathcal{R}_1)\leq f(\mathcal{R}_2)$.
Submodularity means $f(\mathcal{R})$ has diminishing gain: for any $\mathcal{R}_1 \subset \mathcal{R}_2 \subset \mathcal{V}$ and $v \in \mathcal{V} \setminus \mathcal{R}_2$ the discrete derivative $\Delta f(v|\mathcal{R})=f(\mathcal {R} \cup \{v\})-f(\mathcal{R})$ satisfies $\Delta f(v|\mathcal{R}_2)\leq \Delta f(v|\mathcal{R}_1)$.
As will be seen below (see (\ref{eqn_submodular_guarantee})), this implies that
greedy node removal based on LFVC is almost as effective as the combinatorially complex batch algorithm that searches over all possible removal sets $\mathcal{R}$.
\begin{thm}
\label{Thm_submodular}
$f(\mathcal{R})$ is a monotonic submodular set function.
\end{thm}
\begin{proof}
The proof can be found in Appendix \ref{appen_submodular_proof}.
\end{proof}

Based on Theorem \ref{Thm_submodular}, we propose a greedy node-LFVC based node removal algorithm for deep community detection as summarized in Algorithm \ref{algo_LFVC_detection}. 
Algorithm \ref{algo_LFVC_detection} yields an adjacency matrix $\hat{\mathbf{A}}$ that corresponds to the remaining edges after node removal. 
In addition to a list of the $q$ removed nodes, the deep communities are defined by the non-singleton connected components in $\hat{\mathbf{A}}$ supplemented by the nodes that were removed, where the membership of these nodes is defined by the connected components in $\hat{\mathbf{A}}$ to which they connect. More specifically, if $\hat{S}=(\mathcal{V}_{\hat{S}},\mathcal{E}_{\hat{S}})$ denotes one of these non-singleton connected components, 
the set $\mathcal{V}_{\hat{S}} \cup \left\{i \in \mathcal{R}: \mathbf{A}_{ij}=1 \textnormal{~for~some}~j \in \hat{S}\right\}$ is called a deep community.	This definition means that some of the removed nodes may be shared by more than one deep community. 
The following theorem shows that this greedy algorithm has bounded performance loss no worse than $0.63$ as compared with the optimal combinatorial batch removal strategy.

\begin{algorithm}
\caption{Deep Community Detection by greedy node-LFVC}
\label{algo_LFVC_detection}
\begin{algorithmic}
\State \textbf{Input:} Adjacency matrix $\mathbf{A}$, number of removed nodes $q$
\State \textbf{Output:} Deep communities
\State $\mathcal{R}=\varnothing$
\For{$i=1$ to $q$}
    \State Find the largest connected component
    \State Compute the corresponding Fiedler vector $\mathbf{y}$
    \State Find $i^*=\arg \max_{i}  \sum_{j \in \mathcal{N}_i}(y_i-y_j)^2$
    \State $\mathcal{R}=\mathcal{R} \cup i^*$
    \State Remove $i^*$ and its edges from the graph
\EndFor

\State  Find $\hat{S}$,  one of the non-singleton connected components.
\State The set $\mathcal{V}_{\hat{S}} \cup \left\{i \in \mathcal{R}: \mathbf{A}_{ij}=1 \textnormal{~for~some}~j \in \hat{S}\right\}$ is a deep community.
\end{algorithmic}
\end{algorithm}

\begin{thm}
\label{Thm_greedy_guarantee}
Fix the target number of nodes to be removed as $|\mathcal{R}|=q$. Let $\mathcal{R}_{\textnormal{opt}}$ be the optimal node removal set that maximizes $f(\mathcal{R})$ and let $\mathcal{R}_k$ be the greedy node removal set at the $k$-th stage of Algorithm \ref{algo_LFVC_detection}, where $|\mathcal{R}_k|=k$. Then
\begin{align}
f(\mathcal{R}_{\textnormal{opt}})-f(\mathcal{R}_{q})
\leq \left(1-\frac{1}{q}\right)^q f(\mathcal{R}_{\textnormal{opt}}) 
 \leq  \frac{1}{e} f(\mathcal{R}_{\textnormal{opt}}). \nonumber
\end{align}
Furthermore,
\begin{align}
\lambda_2({\widetilde{\mathbf{L}}(\mathcal{R}_{q})})
\leq \lambda_2({\mathbf{L}})-\left(1-e^{-1}\right)f(\mathcal{R}_{\textnormal{opt}}).
\end{align}
\end{thm}
\begin{proof}
The proof can be found in Appendix \ref{appen_greedy_proof}.
\end{proof}

The submodularity of the function $f$ implies that after $q$ greedy iterations the performance loss is within a factor $1/e$ of optimal batch removal \cite{Nemhauser78}. In other words, when removing $\mathcal{R}_{q}$ from $G$, the algebraic connectivity is guaranteed to decrease by at least $(1-e^{-1})f(\mathcal{R}_{\textnormal{opt}})$ of its original value.
Consequently, identifying the top $q$ nodes affecting algebraic connectivity can be regarded as a monotonic submodular set function maximization problem, and the greedy algorithm can be applied iteratively to remove the node with the highest node-LFVC. Similarly, we can use edge-LFVC to detect deep communities by successively remove the edge with the highest edge-LFVC from the graph, and it is easy to show that the term $\sum_{(i,j)\in \mathcal{E}_\mathcal{R}} (y_i-y_j)^2$ in (\ref{eqn_UB_alge_multiple_link}) is a monotonic submodular set function of the edge removal set $\mathcal{E}_\mathcal{R}$.

\section{Deep Community Detection in Stochastic Block Model}
\label{Sec_block}

To demonstrate the effectiveness of using the proposed LFVC for deep community detection, we compare its detection performance to that of other methods for a synthetic network generated by a stochastic block model (SBM) \cite{Holland83,Snijders97,Karrer11}.
Consider a deep community of size $\nin$ embedded in a network of size $n$, $\nin < n$, and let $\nout=n-\nin$ denote the rest of the graph size. The average number of edges between the members of the deep community is denoted by $\cin$, and the average number of edges between the members that are not in the deep community is denoted by $\cout$.
We assume a restricted stochastic block model characterized by the following $2 \times 2$ group connection probability matrix, which specifies the community interconnectivity probabilities in the SBM:
\begin{align}
\label{eqn_SBN}
\mathbf{P} = \bordermatrix{~ & \text{deep} & \text{others} \cr
                  \text{deep} & \pin & \pout \cr
                  \text{others} & \pout & \pout \cr},
\end{align}
where $\pin=\frac{\cin}{\nin}$ and $\pout=\frac{\cout}{\nout}$ are the edge connection probabilities within and outside the deep community, respectively.
According to the definition of community in Sec. \ref{subsec_community}, when $\cin > \cout$, the nodes in the deep community form a community.
The planted clique problem \cite{Alon98} is a special case of the SBM in (\ref{eqn_SBN}) when $\pin=1$. The detection performance of the modularity method has been analyzed in \cite{Nadakuditi12Plant} for the planted clique problem.

The adjacency matrix generated by the SBM in (\ref{eqn_SBN}) is a random binary matrix with partitioned structure
\begin{align}
\label{eqn_asym_block_model}
\mathbf{A} = \begin{bmatrix}
       \Ain & \C           \\
       \C^T           & \Aout
     \end{bmatrix},
\end{align}
where $\Ain$ and $\Aout$ are the adjacency matrices of a Erdos-Renyi graph with edge connection probabilities $\pin$ and $\pout$, respectively.
$\C$ is an $\nin$-by-$\nout$ binary matrix with its entry being $1$ with probability $\pout$.
$\mathbf{A}$  reduces to the special case of (\ref{eqn_Adjacency_Boolean}) when $g=1$. $\mathbf{A}_1 =\left[ \begin{smallmatrix}
                                                                                \Ain & \mathbf{0} \\
                                                                                \mathbf{0} & \mathbf{0} \\
                                                                              \end{smallmatrix} \right]
$ represents within-community connectivity structure and $\mathbf{A}_{\text{nse}}=\left[ \begin{smallmatrix}
                                                                                \mathbf{0} & \C \\
                                                                                \C^T & \Aout \\
                                                                              \end{smallmatrix} \right]$
                                                                              represents the noisy part outside the deep community.

In the following paragraphs we use random matrix theory and concentration inequalities to show that asymptotically the nodes/edges in the noisy part are more likely to have high LFVC. Therefore the removal strategy based on LFVC will with high probability detect the noisy part to reveal the deep community.

Let $\onein$ be the all-ones vector of length $\nin$ and $\oneout$ be the all-ones vector of length $\nout$, and let $\Din=\text{diag}\left(\C\oneout\right)$ and $\Dout=\text{diag}\left(\C^T\onein\right)$.
Following (\ref{eqn_asym_block_model}), the corresponding graph Laplacian matrix can be represented as
\begin{align}
\label{eqn_Laplacian_block}
\mathbf{L} = \begin{bmatrix}
       \Lin+\Din & -\C           \\
       -\C^T           & \Lout+\Dout
     \end{bmatrix},
\end{align}
where $\Lin$ and $\Lout$ are the graph Laplacian matrices of the two Erdos-Renyi graphs.

\begin{thm}
\label{thm_phase}
Let $n=\nin+\nout$ and let $\by=[\yin^T~\yout^T]^T$ be the Fiedler vector of the graph Laplacian matrix $\mathbf{L}$ in (\ref{eqn_Laplacian_block}). Consider the stochastic block model in (\ref{eqn_SBN}). For a fixed $\pin$, there exists an asymptotic threshold $\pout^*$ such that almost surely,
\begin{align}
\left\{
  \begin{array}{ll}
    \sqrt{\frac{n \nin}{\nout}} \yin \ra \pm \onein~\text{and}~ \sqrt{\frac{n \nout}{\nin}} \yout \ra \mp \oneout, & \hbox{\text{if~} $\pout \leq \pout^*$,} \\
    \onein^T \yin \ra 0~\text{and}~\oneout^T \yout \ra 0, & \hbox{\text{if~} $\pout > \pout^*$,}
  \end{array}
\right. \nonumber
\end{align}
as $\nin \ra \infty,~\nout \ra \infty$ and $\frac{\nin}{\nout} \ra c >0$.
\end{thm}

\begin{proof}
The proof can be found in Appendix \ref{appen_phase_proof}.
\end{proof}

Consequently, when $\pout \leq \pout^*$ the
elements of the Fiedler vector tend to have opposite signs within and outside the deep community. Recalling the edge and node LFVC in (\ref{edge_LFVC}) and (\ref{node_LFVC}),
 Theorem \ref{thm_phase} implies that the nodes/edges with high LFVC are more likely to be present on the periphery of  the deep community.
On the other hand, when $\pout > \pout^*$ both $\yin$ and $\yout$ have alternating signs and the deep community cannot be reliably detected by LFVC due to incorrect node/edge removals.
One can interpret $\cin = \nin \cdot \pin$ as the signal strength and $\cout=\nout \cdot \pout$ as the noise level. Based on Theorem \ref{thm_phase}, for a fixed $\cout$ there exists an asymptotic signal strength threshold $\cin^*$ such that the $\yin$ and $\yout$ become constant vectors with opposite signs when $\cin \geq \cin^*$. These results are consistent with the planted clique detection analysis in \cite{Nadakuditi12Plant} that a clique is detectable if the ratio of its within-clique connections to its outside-clique connections is above a certain threshold.

\begin{figure}[t]
	\centering
	\includegraphics[width=3.5in]{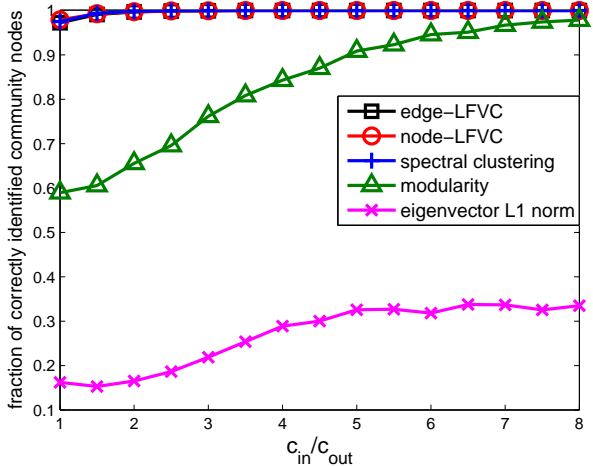}
	\caption{Sensitivity comparisons of community detection algorithms. This figure shows fraction of correctly identified community nodes for simulated stochastic block model (\ref{eqn_asym_block_model}) using (\ref{eqn_SBN}). There is a single deep community, $\nin=40$, $n=200$, and $\cout=2$. The parameter $q$ in the proposed node-LFVC algorithm is selected adaptively: Algorithm \ref{algo_LFVC_detection} stops when the first non-singleton connected component is discovered.  The other algorithms are also implemented with knowledge that there is only one community.    	
		The results are averaged over $100$ trials. The proposed deep community detection algorithm based on LFVC is capable of uncovering the community structure. Spectral clustering has similar performance since it partitions the graph based on the Fiedler vector. The modularity method fails to detect the deep community in the low $\frac{\cin}{\cout}$ regime.}
	\label{Fig_n200m20}
\end{figure}

\begin{figure}[t]
	\centering
	\includegraphics[width=3.5in]{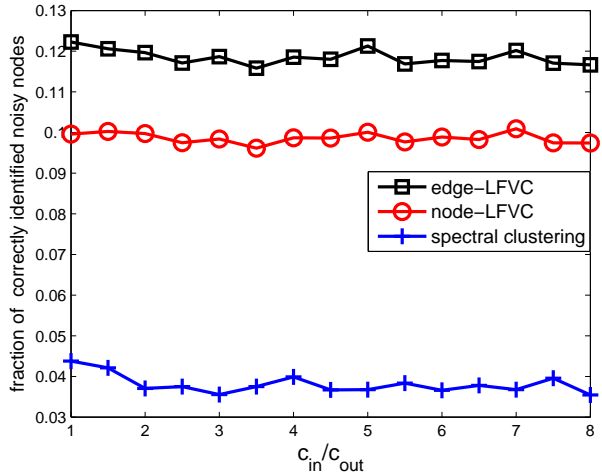}
	\caption{Specificity comparisons of community detection algorithms. For the same algorithms as in Fig. 1, this figure shows the fraction of correctly identified noisy nodes for simulated stochastic block model (\ref{eqn_asym_block_model}) using (\ref{eqn_SBN}). There is a single deep community, $\nin=40$, $n=200$, and $\cout=2$. The results are averaged over $100$ trials. 
		Only the three community detection algorithms having the highest sensitivity (edge-LFVC, node-LFVC, and spectral clustering) are shown here.
		Comparing to Fig. \ref{Fig_n200m20}, when $\frac{\cin}{\cout} \geq 2.5$, the proposed deep community detection algorithm based on LFVC denoises the observed adjacency matrix $\mathbf A$ in (\ref{eqn_Adjacency_Boolean}), removing some noisy edges in $\mathbf{A}_{\text{nse}}$ while retaining the deep community's edges in $\mathbf{A}_{\text{in}}$.
		Spectral clustering is less effective in denoising $\mathbf{A}$. The results validate Theorem \ref{thm_phase} that the Fiedler vector tend to have opposite signs within and outside the deep community when $\cin$ exceeds the threshold $\cin^*=5$.}
	\label{Fig_n200m20_2}
\end{figure}

\begin{figure}[t]
	\centering
	\includegraphics[width=3.5in]{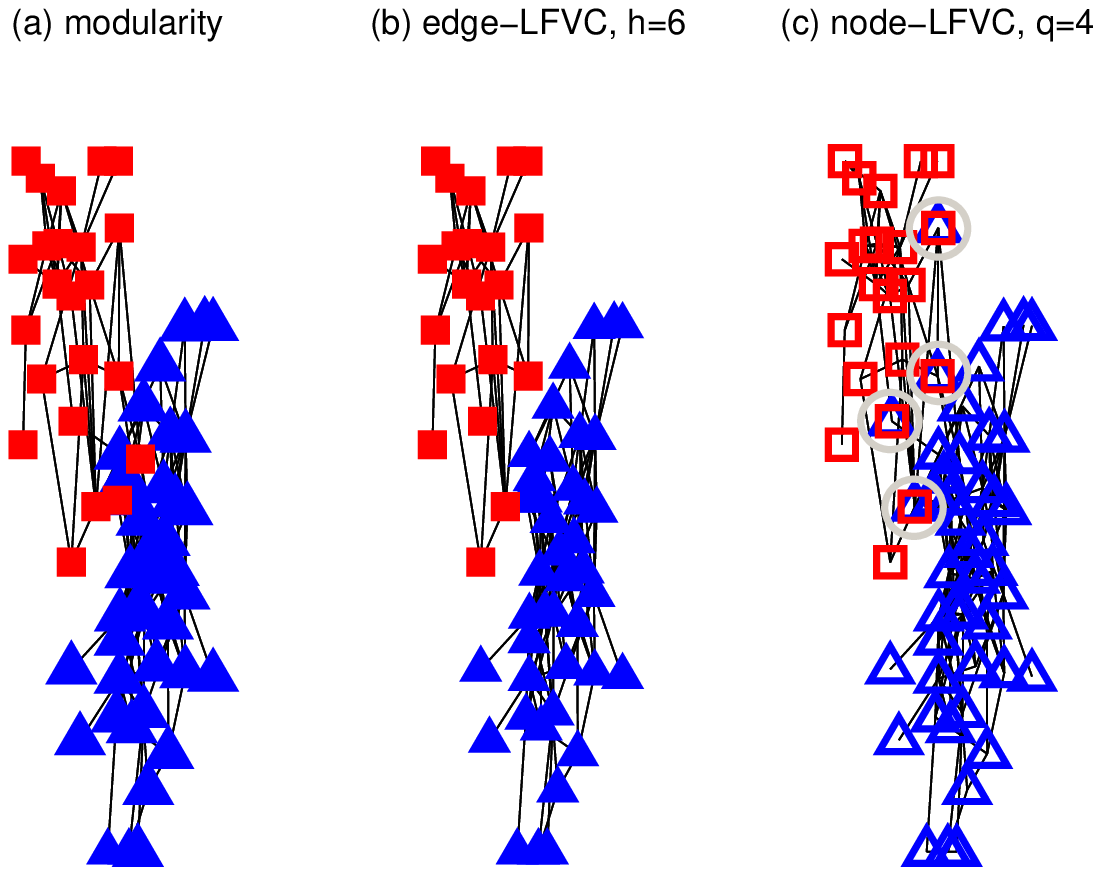}
	\caption{Dolphin social network \cite{Lusseau03Dolphin} with $n=62$ nodes and $m=159$ edges. (a) The modularity method. (b) Edge-LFVC community detection with $h=6$ edge removals. (c) Node-LFVC community detection with $q=4$ node removals. Using node-LFVC, we are able to  identify the four dolphins that interact with two groups as marked by nodes in gray circles. This algorithm, defined by Algorithm \ref{algo_LFVC_detection}, detects that these four nodes are members of the two communities.  The result of spectral clustering is shown in the supplementary file$^1$. Spectral clustering results in the same discovered communities as the proposed  edge-LFVC community detection method. However, unlike the proposed node-LFVC method it does not explicitly identify the four mixed membership dolphins that connect the two communities.}
	\label{Fig_Dolphin}
\end{figure}

The proposed deep community detection method in Algorithm \ref{algo_LFVC_detection} is implemented by sequentially removing the node (edge) with the highest LFVC until the graph becomes disconnected. For comparison, the modularity method is implemented by dividing the network into two communities.
The L1 norm subgraph detection method \cite{MillerICASSP10,Miller10} is also implemented. For the null model (i.e., absence of community structure) of the L1 norm subgraph detection method, we generate $500$ Erdos-Renyi random graphs with edge connection probability $\pout$ and compute the mean and standard deviation of the L1 norm of each eigenvector associated with the modularity matrix.
Let $\text{mean}(i)$ and $\text{std}(i)$ denote the mean and standard deviation of the L1 norm of $i$-th eigenvector in the null model and let $\ell_i$ denote the L1 norm of the $i$-th eigenvector associated with the modularity matrix $\mathbf{B}$. The test statistic is
\begin{align}
t=\min_{i=1,\ldots,n} \frac{\ell_i-\text{mean}(i)}{\text{std}(i)}.
\end{align}
The presence of a dense subgraph is declared if $|t|\geq 2$, which corresponds to $5\%$ false alarm probability \cite{MillerICASSP10,Miller10}. Let $i^*=\arg\min_{i=1,\ldots,n} \frac{\ell_i-\text{mean}(i)}{\text{std}(i)}$. The deep community is identified by selecting the $\nin$ entries of the $i^*$-th eigenvector having the largest magnitude.

Consider detecting a single deep community of size $40$ embedded in a network of size $200$ generated by the stochastic block model (\ref{eqn_SBN}).  
We consider the following definitions of community detection algorithm. Sensitivity and specificity measures: fraction of correctly identified community nodes detected by the algorithm and fraction of correctly identified noisy nodes detected by the algorithm. Let $\hat{S}$ be a community identified by an algorithm and let $S$ denote the true  community embedded in the observed graph. Then the fraction of correctly identified community nodes is defined as 
$\frac{|S \cap \hat{S}|}{\nin}$, where $|S \cap \hat{S}|$ denote the cardinality of the set that contains the nodes belonging to both  $S$ and $\hat{S}$  and $\nin=|S|$. The fraction of correctly identified noisy nodes is defined as $\frac{|S^c \cap \hat{S}^c|}{\nout}$, where $S^c=G / S$ denotes the complement graph of $S$ and $\nout=|S^c|$. Together these sensitivity and specificity measures can be used to assess the performance of deep community detection.

As shown in Fig. \ref{Fig_n200m20}, deep community detection based on LFVC is capable of detecting the embedded community. Spectral clustering is implemented by partitioning the entire graph into two subgraphs $\hat{S}_1$ and $\hat{S}_2$,
	and selecting $\hat{S}= \hat{S}_{i^*}$ as the identified community, where
	$i^*=\arg \max_{i \in \{1,2\}} \{|S \cap \hat{S}_1|,|S \cap \hat{S}_2|\}$. The modularity method is implemented in a similar fashion.

Spectral clustering \cite{Luxburg07} has similar performance to the proposed method since it partitions the graph based on the Fiedler vector.
On the other hand, the modularity method is overly influenced by the noisy part resulting in  degraded community detection. In particular it incorrectly identifies members of the deep community as different groups, especially when $\frac{\cin}{\cout}$ is small. The inaccuracy of the modularity method is due to the fact that in the low  $\frac{\cin}{\cout}$ regime, the network is overly similar to the corresponding degree-equivalent random graph model  used to define the modularity metric. That is, modularity does not provide enough evidence that a community is present.

The fraction of correctly identified noisy nodes via LFVC and spectral clustering is shown in Fig. \ref{Fig_n200m20_2}. 
In this setting, the fraction of identified noisy nodes via node-LFVC is slightly less than that via edge-LFVC. Spectral clustering is less effective in denoising the observed adjacency matrix $\mathbf{A}$.
Observe that the fraction of identified noisy nodes is stable when $\frac{\cin}{\cout} \geq 2.5$. The results validate Theorem \ref{thm_phase} that the elements of the Fiedler vector tend to have opposite signs within and outside the deep community when $\cin$ exceeds the threshold $\cin^*=5$. Consequently, the results show that removal of nodes/edges with high LFVC helps to reveal the true community structure and improve community detection performance.

\section{Deep Community Detection on Real-world Social Network Datasets}
\label{Sec_performance}
In this section, we use the proposed node and edge centrality measures to perform deep community detection on several datasets collected from real-world social networks.
In the implementations of the community detection methods below, the number of removed nodes or edges is a user-specified free parameter. For LFVC (Algorithm \ref{algo_LFVC_detection}) this parameter can be selected based on the bounds established in Theorem \ref{Thm_number_community}.
 We define $h$ the number of edge removals, $q$ the number of node removals and $g$ the number of deep communities.
The results are compared with the modularity method and other node centralities discussed in Sec. \ref{subsec_centrality}.
For data visualization, vertex shapes and colors represent different communities, and edges attached to the removed nodes are retained in the figures in comparison with other methods. Nodes with cross labels (black X labels) are singleton survivors that do not belong to any deep communities using LFVC (Algorithm \ref{algo_LFVC_detection}).

\begin{figure}[t]
	\centering
	\includegraphics[width=3.5in]{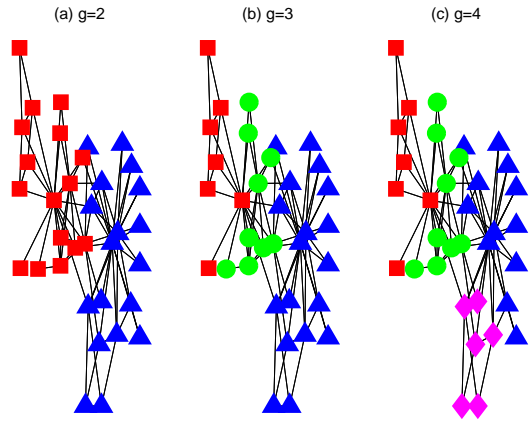}
	\caption{The modularity method on Zachary's karate club \cite{Zachary77Karate} with $n=34$ nodes and $m=78$ edges.}
	\label{Fig_Karate_modularity}
\end{figure}

\begin{figure}[t]
	\centering
	\includegraphics[width=3.5in]{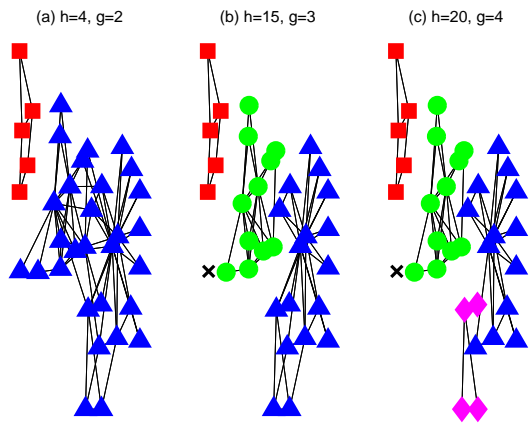}
	\caption{Edge-LFVC community detection on Zachary's karate club \cite{Zachary77Karate} with $n=34$ nodes and $m=78$ edges. For $g=3$ and $4$, the only node with a single acquaintance is excluded from any deep community.}
	\label{Fig_Karate_Edge}
\end{figure}
\begin{figure}[t]
	\centering
	\includegraphics[width=3.5in]{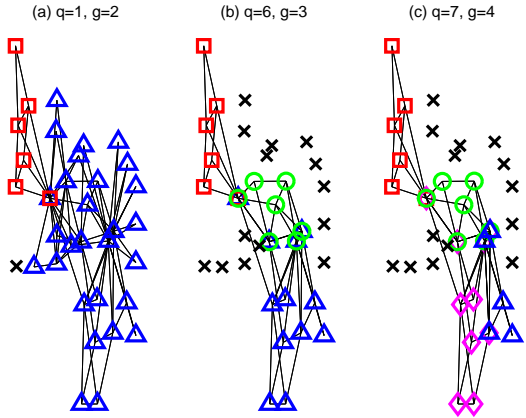}
	\caption{Node-LFVC community detection on Zachary's karate club \cite{Zachary77Karate} with $n=34$ nodes and $m=78$ edges. Important communities and key members are discovered using node-LFVC. This also demonstrates how the singleton survivors (nodes with black X labels) interact through the deep communities. The result of spectral clustering is shown in the supplementary file$^1$. When $g=4$, spectral clustering yields imbalanced communities (one community has single node).}
	\label{Fig_Karate_node}
\end{figure}

\subsection{Dolphin social network}

It is shown in \cite{Lusseau03Dolphin} that there are tight social structures in dolphin populations. Most dolphins interact with other dolphins of the same group and only a few dolphins can interact with dolphins from different groups.
In terms of the proposed LFVC algorithm, these latter Dolphins introduce "noisy" edges connecting the two communities. Figure \ref{Fig_Dolphin} shows that they can therefore be detected by LFVC.
In Fig. \ref{Fig_Dolphin} we compare the results of separating $62$ dolphins into two communities as proposed in \cite{Lusseau03Dolphin}.
For this dataset, community detections based on modularity, edge-LFVC and node-LFVC have high concordance on the community structures. To partition the graph into two communities, we need to remove 6 edges based on edge-LFVC or remove 4 nodes based on node-LFVC. The four dolphins that are able to communicate between these two communities are further identified by node-LFVC.

\subsection{Zachary's karate club}

Zachary's karate club \cite{Zachary77Karate} is a widely used example for social network analysis, which contains interactions among $34$ karate students. Based on the student activities, Zachary determines the ground-truth community structure for $g=2$, which coincides with the result of the modularity method in Fig. \ref{Fig_Karate_modularity} (a). However, the visualization indicates that there are some deep communities embedded in these two communities, such as the five-node community in the upper left corner. Indeed, the modularity will keep increaseing if we further divide communities into $3$ and $4$ small communities as shown in Fig. \ref{Fig_Karate_Edge} (b) and (c), respectively. 

As shown in Fig. \ref{Fig_Karate_Edge} (a), using edge-LFVC, the five-node community in the left upper corner is revealed when we partition the graph into two connected subgraphs. In Fig. \ref{Fig_Karate_Edge} (b), three communities are revealed and the only node with a single acquaintance is excluded from any deep community. Excluding this node makes the community structure more tightly connected compared with Fig. \ref{Fig_Karate_modularity} (b). For $g=4$, the community structure in Fig. \ref{Fig_Karate_Edge} (c) much resembles Fig. \ref{Fig_Karate_modularity} (c) except that we exclude the node having a single acquaintance.

Using node-LFVC, we are able to extract important communities and key members as shown in Fig. \ref{Fig_Karate_node}. For $g=2$, only one node removal is required to partition the graph into two connected subgraphs, which implies that this node is common to the two communities according to the proposed Algorithm \ref{algo_LFVC_detection}. For $g=3$, two deep communities (green circle and blue triangle) are discovered in the largest community (the blue triangle community in Fig. \ref{Fig_Karate_node} (a)), where these two deep communities have dense internal connections compared with the external connections to other members in the largest community.
These discovered deep communities are important communities embedded in the network since they play an important role in connecting the singleton survivors indicated by black X labels. Similar observations hold for $g=4$ in Fig. \ref{Fig_Karate_node} (c).

\subsection{Coauthorship among network scientists}

\begin{figure}[t]
	\centering
	\includegraphics[width=3.5in]{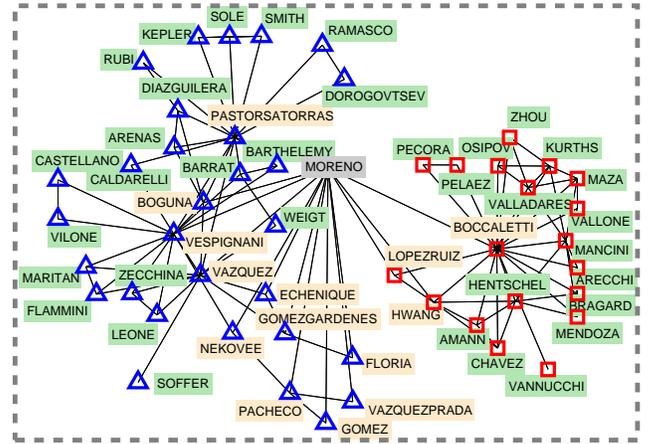}
	\caption{Yamir Moreno's local $2$-hop coauthorship network (from part of the network of coauthorship among network scientists \cite{Newman06community} having $n=379$ nodes and $m=914$ edges). Moreno has $14$ coauthors (marked by light orange color) and his coauthors have $35$ coauthors. The modularity method \cite{Newman06community} detects that Moreno is a member of only one large community (dashed box in gray). The proposed LFVC method detects Moreno as belonging to two separate communities indicated by red and blue nodes, respectively.}
	\label{Fig_NetworkScienceMoreno_twohop_and_modularity}
\end{figure}

\begin{figure}[t]
	\centering
	\includegraphics[width=3.5in]{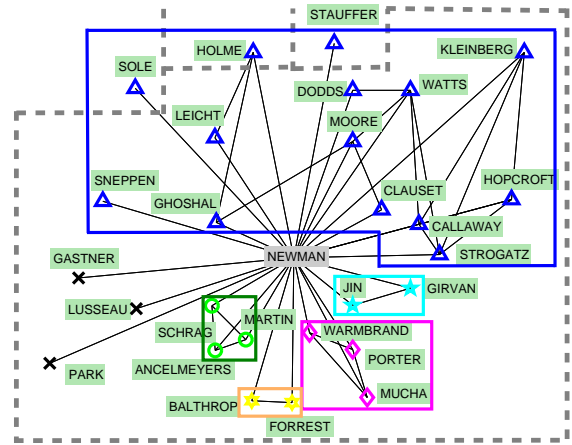}
	\caption{Mark Newman's local $1$-hop coauthor network in the network scientist coauthorship graph \cite{Newman06community}. The proposed LFVC method detects Newman as belonging to $5$ communities (marked by different vertex shapes and colors in solid boxes) and being associated with $3$ singleton survivors (marked by black X label). Notably, Lusseau is detected as singleton survivor since his research area is primarily in zoology.
		As shown in gray dashed box, the modularity method \cite{Newman06community} detects $25$ out of $28$ scholars as being in a single community, and the top left $3$ scholars as belonging to $3$ different communities.}
	\label{Fig_NetworkScienceNewman_and_modularity_2}
\end{figure}

We next examine the coauthorship network studied by Newman \cite{Newman06community}. Nodes represent network scientists and edges represent the existence of coauthorship. Multiple memberships are expected to occur in this dataset since a network scientist may collaborate with other network scientists across different regions all the while having many collaborations with his/her colleagues and students at the same institution. As a result, one would expect, as implemented by Algorithm \ref{algo_LFVC_detection}, 
node-LFVC to be advantageous for identifying authors who with multiple memberships and detecting deep communities.

As shown in Fig. \ref{Fig_NetworkScienceMoreno_twohop_and_modularity}, the first node with the highest node-LFVC is Yamir Moreno, who is a network scientist in Spain but has many collaborators outside Spain. The local (two-hop) coauthorship network of Yamir Moreno is shown in Fig. \ref{Fig_NetworkScienceMoreno_twohop_and_modularity}. The red square community represents the network scientists in Spain and Europe, whereas the blue triangle community represents the rest of the network scientists.

After removing Yamir Moreno from the network, the node with the highest node-LFVC in the remaining largest community is Mark Newman, who is associated with $5$ community memberships and $3$ singleton survivors as shown in Fig. \ref{Fig_NetworkScienceNewman_and_modularity_2}. Each community can be related to certain relationship such as colleagues, students and research institutions.
Notably, Lusseau is detected as a singleton survivor in the deep community detection process in Fig. \ref{Fig_NetworkScienceNewman_and_modularity_2}.
This can be explained by the fact that although Lusseau has coauthorship with Newman, his research area is primarily in zoology and he has no interactions with other network scientists in the dataset since other network scientists are mainly specialists in physics. Also note that the modularity method (gray dashed box) fails to detect these deep communities and it detects 25 out of 28 network scientists in Fig. \ref{Fig_NetworkScienceNewman_and_modularity_2} as one big community.

\subsection{Last.fm online music system}
Last.fm is an online music system which allows users to tag their favorite songs and artists and make friends with other users. We use the friendship dataset collected in \cite{LastFm} for deep community detection based on node-LFVC and the other centralities introduced in Sec. \ref{subsec_centrality}.
Two quantities, the normalized largest community size and the number of discovered communities with respect to node removals, are used to evaluate the performance of community detection when different node centralities are applied.
These two quantities reflect the effectiveness of graph partitioning.
The number of removed nodes is the number of stages for performing deep community detection and removing more nodes reveals more deep communities and key members in the network.

\begin{figure}[t]
    \centering
    \includegraphics[width=3.5in]{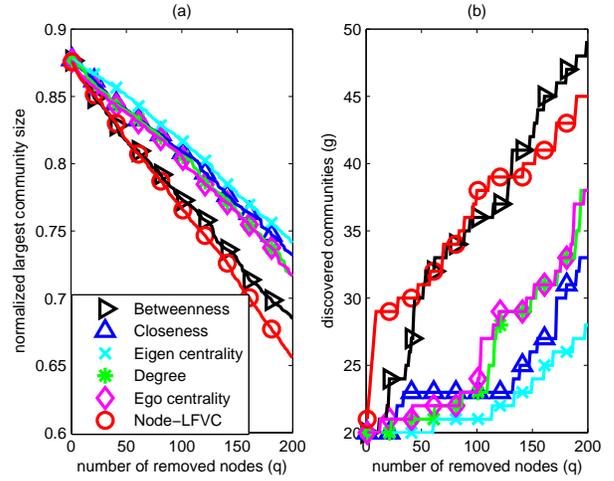}
    \caption{Friendship in Last.fm online music system \cite{LastFm} with $n=1843$ nodes and $m=12668$ edges. (a) Normalized largest community size decreases in the number of node removals at different rates under different node centralities. (b) Discovered communities with respect to node removals using different node centralities. Node-LFVC outperforms other node centralities in terms of minimizing the largest community size, and while being capable of detecting more communities in the network for the first 50 removals.}
    \label{Fig_LastFm}
\end{figure}

\begin{figure}[t]
    \centering
    \includegraphics[width=3.5in]{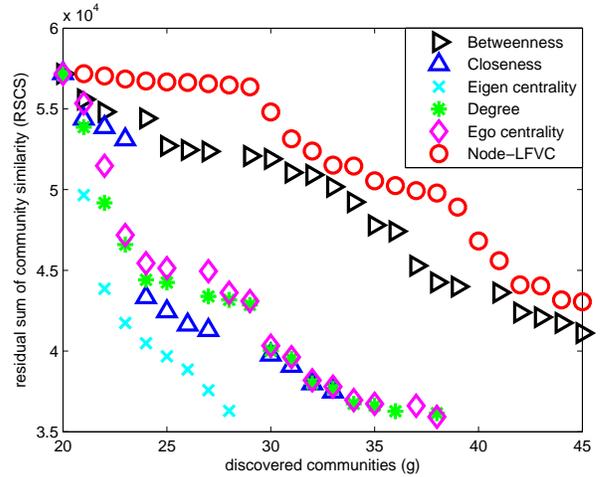}
    \caption{Residual sum of community similarity (RSCS) in Last.fm network. The residual sum of community similarity based on node-LFVC outperforms other centralities, which indicates that node removals based on node-LFVC can best detect deep communities that share common interest in artists.}
    \label{Fig_LastFm_similarity_sum_non_overlap}
\end{figure}

As shown in Fig. \ref{Fig_LastFm} (a), the normalized largest community size decays linearly with respect to the number of node removals. Among all node centralities, node-LFVC has the steepest decaying rate.  Furthermore, using node-LFVC discovers more deep communities, as shown in Fig. \ref{Fig_LastFm} (b) during the first $50$ node removals. The only node centrality that is comparable to node-LFVC is betweenness centrality.

To validate the effectiveness of deep community detection, we use the user-artists dataset in \cite{LastFm} to compute the listening similarity in each discovered community. The dataset contains $17632$ artists and records the number of times each user has listened to an artist. Let $\mathbf{w}_i$ be a $17632$-by-$1$ vector with its $j$-th entry being the number of times the $i$-th user has listened to the $j$-th artist. The residual community similarity (RCS) is defined as the sum of cosine similarity between each user in the same community excluding the nodes that have been removed and the singleton survivors.  The residual community similarity of a deep community $C_k$ is defined as
\begin{align}
\text{RCS}(C_k)=\sum_{i \in {C_k},i \notin \mathcal{R}} \sum_{j \in {C_k},j>i,j \notin \mathcal{R}} \frac{\mathbf{w}_i ^T \mathbf{w}_j}{\|\mathbf{w}_i\|_2 \|\mathbf{w}_i\|_2}.
\end{align}
The residual sum of community similarity (RSCS) is defined as the sum of RCS of each discovered community. That is,
\begin{align}
\text{RSCS}=\sum_{k=1}^g \text{RCS}(C_k).
\end{align}

As shown in Fig. \ref{Fig_LastFm_similarity_sum_non_overlap}, the residual sum of community similarity based on node-LFVC is larger than that for other centralities. This suggests that node removals based on node-LFVC can best detect friendship communities that share common interest in artists. Note that although betweenness may detect more communities in Fig. \ref{Fig_LastFm} (b), Fig. \ref{Fig_LastFm_similarity_sum_non_overlap} shows that the residual sum of community listening similarity based on betweenness is smaller than that based on node-LFVC, which indicate that node-LFVC reveals more accurate community structure than betweenness. The residual sum of community similarity decreases with respect to the number of discovered communities due to the fact that the removed nodes and singleton survivors are excluded  for similarity computation.

\section{conclusion}
\label{sec_con}
Based on bounds on the sensitivity of algebraic connectivity to node or edge removals, we proposed a centrality measure called local Fiedler vector centrality (LFVC) for deep community detection. We proved that LFVC relates to a monotonic submodular set function such that greedy node removals based on LFVC can be applied to identify the most vulnerable nodes or edges with bounded performance loss compared to the optimal combinatorial batch removal strategy. Asymptotic analysis of the Fiedler vector established that LFVC can successfully remove the noisy part while retaining the deep community structure in networks generated by the stochastic block model.
In comparison to the modularity method \cite{Newman06PNAS} and the L1 norm subgraph detection method \cite{Miller10}, we show that LFVC can achieve better community detection performance in correctly identifying the embedded deep communities.

The proposed method provides better resolution for discovering important communities and key members in the real-world social network datasets studied here.
In particular, for the Last.fm online music system dataset, LFVC is shown to significantly outperform other centralities for deep community detection in terms of the residual sum of community listening similarity.
This new measure can likely offer new insights on community structure in other social, biological and technological networks.

\begin{appendices}

\section{Proof of Theorem \ref{Thm_number_community}}
\label{appen_num_community_proof}
From (\ref{eqn_connectivity}) a graph is connected if and only if the algebraic connectivity is greater than zero. Furthermore, the smallest eigenvalue of the associated graph Laplacian matrix is always $0$. Therefore $n-q-r$ is the number of connected components (including the singleton nodes) in $\widetilde{G}$ \cite{Chung97SpectralGraph} by the fact that $n-q$ and $r$ are the node size and rank of $\widetilde{\mathbf{L}}$, respectively. Since the definition of a deep community excludes singleton nodes, the first inequality in (\ref{eqn_UB_of_deep_community}) becomes equality if all connected components in $\widetilde{G}$ are non-singleton.

Using a well-known matrix norm inequality \cite{HornMatrixAnalysis} that
$\|\mathbf{M}\|_\ast \leq r\|\mathbf{M}\|_2$ for any square matrix $\mathbf{M}$ of rank $r$, where $\|\mathbf{M}\|_2=\max_{\|\mathbf{x}\|_2=1} \|\mathbf{M}\mathbf{x}\|_2=\lambda_n(\mathbf{M})$.
We have
\begin{align}
n-q-r \leq n-q-\frac{\|\widetilde{\mathbf{L}}\|_\ast}{\lambda_n(\widetilde{\mathbf{L}})}
=n-q-\frac{2\widetilde{m}}{\lambda_n(\widetilde{\mathbf{L}})}, \nonumber
\end{align}
where $\|\widetilde{\mathbf{L}}\|_\ast = \textnormal{trace}(\widetilde{\mathbf{L}})=2 \widetilde{m}$ is the total degree of $\widetilde{G}$.

Next we show that the second inequality in (\ref{eqn_UB_of_deep_community}) becomes an equality if each non-singleton connected graph is a complete subgraph of the same size.
Consider a graph consisting of $g$ disjoint complete subgraphs of $n^\prime \geq 2$ nodes and $n^\prime(n^\prime-1)/2$ edges. The largest eigenvalue of each subgraph is $n^\prime$ and $\|\widetilde{\mathbf{L}}\|_\ast=g \cdot n^\prime(n^\prime-1)$. The upper bound becomes
$g \cdot n^\prime - \frac{g n^\prime(n^\prime-1)}{n^\prime}=g$, which is exactly the number of non-singleton connected components in $\widetilde{G}$. These results can be directly applied to edge removals in $G$ by setting $q=0$ since no nodes are removed.

\section{Proof of Theorem \ref{Thm_largest_component_size}}
\label{appen_largest_component_proof}
Let $r$ be the rank of $\widetilde{\mathbf{L}}$. We prove that there exists an $n \times (n-r)$ binary matrix $\mathbf{X}=[\mathbf{x}_1~\mathbf{x}_2 \ldots \mathbf{x}_{n-r}]$ whose columns $\{\mathbf{x}_i\}_{i=1}^{n-r}$ satisfy: 1) $\|\mathbf{x}_i\|_1$ is the size of the $i$-th connected component of $\widetilde{G}$; 2) they are orthogonal; 3) they span $\textnormal{null}(\widetilde{\mathbf{L}})$.
Assume $\widetilde{G}$ consists of $g$ connected components. Then there exits a matrix permutation (node relabeling) such that
\begin{align}
\widetilde{\mathbf{L}}=
\left[
  \begin{matrix}
       \widetilde{\mathbf{L}}_1 & 0     & 0      & 0  \\
       0     & \widetilde{\mathbf{L}}_2 & 0      & 0 \\
       0     & 0     & \ddots & 0   \\
       0     & 0     & 0      & \widetilde{\mathbf{L}}_g 
  \end{matrix}
\right].
\end{align}
Associated with the $i$-th block matrix $\widetilde{\mathbf{L}}_i$ we define $\mathbf{x}_i$ as an $n \times 1$ binary vector $\mathbf{x}_i$ in \textnormal{null}$(\widetilde{\mathbf{L}})$ having the form
$\mathbf{x}_i=[0 \ldots 0~1 \ldots 1~0 \ldots 0]^T$,
where the locations of the nonzero entries correspond to the indexes of the $i$-th block matrix. It is obvious that $\| \mathbf{x}_i \|_1=\sum_{j=1}^{n}|x_{ij}|$ equals the size of the $i$-th component and such $\{\mathbf{x}_i\}_{i=1}^{n-r}$ are mutually orthogonal. Furthermore, there exists no other binary matrix which is sparser than $\mathbf{X}$ with column span equal to \textnormal{null}$(\widetilde{\mathbf{L}})$. If there existed another binary matrix that were sparser than $\mathbf{X}$, then it would contradict the fact that its column vectors have sums equal to the component sizes of $\widetilde{G}$. Therefore the largest non-singleton connected component size of $\widetilde{G}$ is $\psi(\widetilde{G})=\| \mathbf{X}\|_1=\max_i \|\mathbf{x}_i\|_1$.

\section{Proof of Lemma \ref{Lemma_nonnegative}}
\label{appen_nonnegative_proof}
By the relation
\begin{align}
\label{eqn_alge_quadratic_neighbor}
\sum_{i \in \mathcal{R}}\sum_{j \in \mathcal{V}}\mathbf{A}_{ij} (y_i-y_j)^2=\sum_{i \in \mathcal{R}}\sum_{j \in \mathcal{N}_i} (y_i-y_j)^2
\end{align}
and $\mathcal{V} = \{\mathcal{V}/\mathcal{R}\} \cup \{\mathcal{R}\}$, we have
\begin{align}
f(\mathcal{R})&=\sum_{i \in \mathcal{R}}\sum_{j\in \mathcal{N}_i} (y_i-y_j)^2-\frac{1}{2} \sum_{i \in \mathcal{R}}\sum_{j \in \mathcal{V}}\mathbf{A}_{ij} (y_i-y_j)^2 \nonumber\\
&~~~+ \frac{1}{2} \sum_{i \in \mathcal{R}}\sum_{j \in \mathcal{V}/\mathcal{R}} \mathbf{A}_{ij} (y_i-y_j)^2 \nonumber\\
&=\frac{1}{2} \sum_{i \in \mathcal{R}}\sum_{j\in \mathcal{N}_i} (y_i-y_j)^2
+ \frac{1}{2} \sum_{i \in \mathcal{R}}\sum_{j \in \mathcal{V}/\mathcal{R}} \mathbf{A}_{ij} (y_i-y_j)^2 \nonumber \\
&\geq 0.
\end{align}
$f(\varnothing)=0$ follows directly from the definition of $f(\mathcal{R})$ in (\ref{eqn_fR}).

\section{Proof of Theorem \ref{Thm_submodular}}
\label{appen_submodular_proof}
We first prove the monotonic property.
Consider two node removal sets $\mathcal{R}_1 \subset \mathcal{R}_2 \subset \mathcal{V}$. Then using Lemma \ref{Lemma_nonnegative} and the fact that $\mathcal{R}_1 / \mathcal{R}_2 =\varnothing$,
\begin{align}
\label{eqn_UB_onotone}
&f(\mathcal{R}_2)-f(\mathcal{R}_1) \nonumber \\
&=\sum_{i \in \mathcal{R}_2/\mathcal{R}_1} \sum_{j \in \mathcal{N}_i} (y_i-y_j)^2-\sum_{i \in \mathcal{R}_1} \sum_{j \in \mathcal{R}_2/\mathcal{R}_1} \mathbf{A}_{ij}(y_i-y_j)^2 \nonumber \\
&~~~-\frac{1}{2} \sum_{i \in \mathcal{R}_2/\mathcal{R}_1} \sum_{j \in \mathcal{R}_2/\mathcal{R}_1} \mathbf{A}_{ij}(y_i-y_j)^2 \nonumber \\
& = \sum_{i \in \mathcal{R}_2/\mathcal{R}_1} \sum_{j \in \mathcal{V}} \mathbf{A}_{ij} (y_i-y_j)^2-\sum_{i \in \mathcal{R}_1} \sum_{j \in \mathcal{R}_2/\mathcal{R}_1} \mathbf{A}_{ij}(y_i-y_j)^2 \nonumber \\
&~~~-\sum_{i \in \mathcal{R}_2/\mathcal{R}_1} \sum_{j \in \mathcal{R}_2/\mathcal{R}_1} \mathbf{A}_{ij}(y_i-y_j)^2  \nonumber \\
&~~~+\frac{1}{2}\sum_{i \in \mathcal{R}_2/\mathcal{R}_1} \sum_{j \in \mathcal{R}_2/\mathcal{R}_1} \mathbf{A}_{ij}(y_i-y_j)^2
\nonumber \\
& = \sum_{i \in \mathcal{R}_2/\mathcal{R}_1} \left( \sum_{j \in \mathcal{V}} \mathbf{A}_{ij}(y_i-y_j)^2 - \sum_{j \in \mathcal{R}_2} \mathbf{A}_{ij}(y_i-y_j)^2\right) \nonumber \\
&~~~+\frac{1}{2}\sum_{i \in \mathcal{R}_2/\mathcal{R}_1} \sum_{j \in \mathcal{R}_2/\mathcal{R}_1} \mathbf{A}_{ij}(y_i-y_j)^2 \nonumber \\
& =\sum_{i \in \mathcal{R}_2/\mathcal{R}_1}
\sum_{j \in \mathcal{V}/\mathcal{R}_2} \mathbf{A}_{ij}(y_i-y_j)^2
\nonumber \\
&~~~+\frac{1}{2}\sum_{i \in \mathcal{R}_2/\mathcal{R}_1} \sum_{j \in \mathcal{R}_2/\mathcal{R}_1} \mathbf{A}_{ij}(y_i-y_j)^2 \nonumber \\
& \geq 0.
\end{align}
Therefore $f(\mathcal{R})$ is a monotonic increasing set function (i.e., $f(\mathcal{R}_2) \geq f(\mathcal{R}_1)$ for all $\mathcal{R}_1 \subset \mathcal{R}_2 \subset \mathcal{V}$).

Furthermore, $f(\mathcal{R})$ is a submodular set function \cite{Nemhauser78,Fujishige90} since for any node $v \in \mathcal{V}, v \notin \mathcal{R}_2$, $\mathcal{R}_1 \subset \mathcal{R}_2 \subset \mathcal{V}$, we have from (\ref{eqn_fR}) that
\begin{align}
f(\mathcal{R}_1 \cup v) - f(\mathcal{R}_1) &= \sum_{j\in \mathcal{N}_v} (y_v-y_j)^2- \sum_{j \in \mathcal{R}_1}\mathbf{A}_{vj} (y_v-y_j)^2 \nonumber \\
& \geq \sum_{j\in \mathcal{N}_v} (y_v-y_j)^2- \sum_{j \in \mathcal{R}_2}\mathbf{A}_{vj} (y_v-y_j)^2 \nonumber \\
& = f(\mathcal{R}_2 \cup v) - f(\mathcal{R}_2).
\end{align}
This diminishing returns property of $f(\mathcal{R})$ establishes that $f$ is submodular \cite{krause2012submodular}.

\section{Proof of Theorem \ref{Thm_greedy_guarantee}}
\label{appen_greedy_proof}
By submodularity of $f(\mathcal{R})$ in Theorem \ref{Thm_submodular}, there exists a $v \in \mathcal{R}_{\textnormal{opt}} / \mathcal{R}_k$ \cite{Fujishige90} such that
\begin{align}
f(\mathcal{R}_k \cup v)-f(\mathcal{R}_k) \geq \frac{1}{q} \left( f(\mathcal{R}_{\textnormal{opt}})-f(\mathcal{R}_k)\right).
\end{align}
After algebraic manipulation, we have
\begin{align}
f(\mathcal{R}_{\textnormal{opt}})-f(\mathcal{R}_{k+1}) \leq \left(1-\frac{1}{q}\right)(f(\mathcal{R}_{\textnormal{opt}})-f(\mathcal{R}_k))
\end{align}
and therefore
\begin{align}
f(\mathcal{R}_{\textnormal{opt}})-f(\mathcal{R}_{q}) \leq \left(1-\frac{1}{q}\right)^q f(\mathcal{R}_{\textnormal{opt}}) \leq \frac{1}{e} f(\mathcal{R}_{\textnormal{opt}}).
\end{align}
Applying this result to (\ref{eqn_UB_multiple_node}), we have
\begin{align}
\label{eqn_submodular_guarantee}
\lambda_2({\widetilde{\mathbf{L}}(\mathcal{R}_{q})}) &\leq \lambda_2({\mathbf{L}})-f(\mathcal{R}_{q})  \nonumber \\
& \leq \lambda_2({\mathbf{L}})-\left(1-e^{-1}\right)f(\mathcal{R}_{\textnormal{opt}}).
\end{align}

\section{Proof of Theorem \ref{thm_phase}}
\label{appen_phase_proof}
Let $\mathbf{x}=[\xin^T~\xout^T]^T$. By (\ref{eqn_alge}) we have
$\lambda_2(\mathbf{L})=\min_{\bx} \bx^T \mathbf{L} \bx$ subject to $\xin^T\xin+\xout^T\xout=1$ and $\xin^T\onein+\xout^T\oneout=0$. Using Lagrange multipliers $\mu$, $\nu$ and (\ref{eqn_Laplacian_block}), the Fiedler vector $\by=[\yin^T~\yout^T]^T$ of $\mathbf{L}$ is a minimizer of the function over $\mathbf{x}$:
\begin{align}
\label{eqn_Lagrange}
\Gamma&=\xin^T (\Lin+\Din) \yin + \xout^T (\Lout+\Dout) \xout -2\xin^T \C \xout \nonumber \\
 &~~-\mu(\xin^T\xin+\xout^T\xout-1)-\nu (\xin^T\onein+\xout^T\oneout).
\end{align}
Differentiating (\ref{eqn_Lagrange}) with respect to $\xin$ and $\xout$ respectively, and substituting $\by$ to the equations,
\begin{align}
\label{eqn_Lagrange1}
&2(\Lin+\Din) \yin -2\C \yout - 2\mu \yin -\nu \onein=\mathbf{0}; \\
\label{eqn_Lagrange2}
&2(\Lout+\Dout) \yout -2\C^T \yin - 2\mu \yout -\nu \oneout=\mathbf{0}.
\end{align}
Multiplying $\onein^T$ to (\ref{eqn_Lagrange1}) and $\oneout^T$ to (\ref{eqn_Lagrange2}) from the left hand side, we have
\begin{align}
\label{eqn_Lagrange3}
&2\onein^T\Din \yin -2\onein^T\C \xout - 2\mu \onein^T\xin -\nu \nin=0; \\
\label{eqn_Lagrange4}
&2\oneout^T\Dout \xout -2\oneout^T\C^T \xin - 2\mu \oneout^T\xout -\nu \nout=0.
\end{align}
Since $\onein^T\Din=\oneout^T\C^T$ and $\onein^T\C=\oneout^T\Dout$, summing (\ref{eqn_Lagrange3}) and (\ref{eqn_Lagrange4}) we obtain for the Lagrange multiplier $\nu$:
\begin{align}
\nu=-\frac{2\mu}{n} (\yin^T\onein+\yout^T\oneout)=0
\end{align}
due the fact that $\by \perp \mathbf{1}$.
Applying $\nu=0$ and left multiplying  (\ref{eqn_Lagrange1}) by $\yin^T$ and  (\ref{eqn_Lagrange2}) by  $\yout^T$, we have
\begin{align}
\label{eqn_Lagrange5}
&\yin^T (\Lin+\Din) \yin -\yin^T \C \yout-\mu \yin^T \yin =0;\\
\label{eqn_Lagrange6}
&\yout^T (\Lout+\Dout) \yout -\yout^T \C^T \yin-\mu \yout^T \yout =0.
\end{align}
Since $\by$ is the Fiedler vector, summing (\ref{eqn_Lagrange5}) and (\ref{eqn_Lagrange6}) together we obtain
\begin{align}
\label{eqn_mu}
\mu=\lambda_2(\mathbf{L}).
\end{align}

Let $\bCbar=\pout \onein \oneout^T$, where its entry is the mean of an i.i.d Bernoulli random variable of an entry in $\C$. Let $\sigma_i(\mathbf{M})$ denote the $i$-th largest singular value of $\mathbf{M}$ and write $\C=\bCbar+\C-\bCbar=:\bCbar+\mathbf{\Delta}$. Since $\mathbf{\Delta}_{ij}=1-\pout$ with probability $\pout$ and $\mathbf{\Delta}_{ij}=-\pout$ with probability $1-\pout$,
by Latala's theorem \cite{Latala05}, $\mathbb{E}\sigma_1(\mathbf{\Delta}/\sqrt{\nin \nout})\rightarrow 0$ as $\nin, \nout \ra \infty$.  Therefore by Talagrand's concentration inequality \cite{Talagrand95}, the singular values of $\C/\sqrt{\nin \nout} $ converge to $\pout$, i.e.,
$\sigma_1(\C/\sqrt{\nin \nout} ) \overset{\text{a.s.}}{\longrightarrow} \sigma_1(\bCbar/\sqrt{\nin \nout} )=\pout$ and $\sigma_i(\C/\sqrt{\nin \nout} ) \overset{\text{a.s.}}{\longrightarrow} 0$ for $i \geq 2$ when $\nin \rightarrow \infty$, $\nout \rightarrow \infty$.
We further assume $\nin$ and $\nout$ grow with a constant rate so that $\frac{\nin}{\nout} \rightarrow c$, where $c$ is a positive constant.
Furthermore, as proved in \cite{BenaychGeorges12}, the left/right singular vectors of $\C$ and $\bCbar$ are close to each other in the sense that the squared inner product of their left/right singular vectors converges to $1$ almost surely when $\sqrt{\nin \nout} \pout \rightarrow \infty$.
Consequently, we have $\frac{1}{\nout}\Din \onein=\frac{1}{\nout}\C \oneout \ra \pout \onein$ and $\frac{1}{\nin}\Dout \oneout=\frac{1}{\nin}\C^T \onein \ra \pout \oneout$ almost surely.

Applying these results to (\ref{eqn_Lagrange3}) and (\ref{eqn_Lagrange4}), we have, almost surely,
\begin{align}
\label{eqn_Lagrange7}
& \frac{\pout \onein^T \yin}{\sqrt{c}}-\sqrt{c} \pout \oneout^T \yout -\frac{\mu \onein^T \yin}{\sqrt{\nin\nout}} \ra 0;\\
\label{eqn_Lagrange8}
&\sqrt{c} \pout \oneout^T \yout-\frac{\pout \onein^T \yin}{\sqrt{c}} - \frac{\mu \oneout^T \yout}{\sqrt{\nin\nout}}  \ra 0.
\end{align}
By the fact that $\onein^T \yin+\oneout^T \yout=0$, we have, almost surely,
\begin{align}
\label{eqn_Lagrange9}
&\lb \sqrt{c}+\frac{1}{\sqrt{c}}\rb\left(\pout-\frac{\mu}{n} \right) \onein^T \yin \ra 0;\\
\label{eqn_Lagrange10}
&\lb \sqrt{c}+\frac{1}{\sqrt{c}}\rb\left(\pout-\frac{\mu}{n} \right) \oneout^T \yout \ra 0.
\end{align}
Consequently, recalling $\mu=\lambda_2(L)$ in (\ref{eqn_mu}), at least one of the two cases has to be satisfied:
\begin{align}
\label{eqn_Lagrange11}
& \text{Case 1:}~\frac{\lambda_2(\mathbf{L})}{n}\overset{\text{a.s.}}{\longrightarrow}\pout; \\
\label{eqn_Lagrange12}
& \text{Case 2:}~\onein^T \yin \overset{\text{a.s.}}{\longrightarrow} 0~\text{and}~\oneout^T \yout \overset{\text{a.s.}}{\longrightarrow} 0.
\end{align}

Similar to the results in \cite{Radicchi13}, the algebraic connectivity and the Fiedler vector undergo an asymptotic structural transition between Case 1 and Case 2. That is, a transition from Case 1 to Case 2 occurs when $\pout$ exceeds a certain threshold $\pout^*$.
In Case 1, the asymptotic algebraic connectivity grows linearly with $\pout$. Furthermore, from (\ref{eqn_Lagrange5}), (\ref{eqn_Lagrange6}), (\ref{eqn_Lagrange11}) and $\onein^T \yin+\oneout^T \yout=0$,
in Case 1 the Fiedler vector $\by$ has the following property: almost surely,
\begin{align}
\label{eqn_Lagrange13}
&\frac{\yin^T \Lin \yin}{\sqrt{\nin \nout}} + \frac{\pout}{\sqrt{\nin \nout}}(\onein^T \yin)^2-\sqrt{c}\yin^T\yin \ra 0;\\
\label{eqn_Lagrange14}
&\frac{\yout^T \Lout \yout}{\sqrt{\nin \nout}} + \frac{\pout}{\sqrt{\nin \nout}}(\onein^T \yin)^2-\frac{\yout^T\yout}{\sqrt{c}} \ra 0.
\end{align}
Summing them up, we have
\begin{align}
\label{eqn_Lagrange15}
 &\left[ \frac{2(\onein^T \yin)^2}{\sqrt{\nin \nout}}
- \left(\sqrt{c}\yin^T\yin+\frac{\yout^T\yout}{\sqrt{c}}\right) \right]\pout  \nonumber\\
&~~~+\frac{1}{\sqrt{\nin \nout}}\left(\yin^T \Lin \yin + \yout^T \Lout \yout \right)
\overset{\text{a.s.}}{\longrightarrow}0.
\end{align}
Since in Case 1 the asymptotic Fiedler vector $\by$ is the same for different values of $\pout$ when $\pout \leq \pout^*$, this implies almost surely,
\begin{align}
\label{eqn_Lagrange16}
&\frac{1}{\sqrt{\nin \nout}}\left(\yin^T \Lin \yin + \yout^T \Lout \yout \right) \ra 0; \\
&\frac{2(\onein^T \yin)^2}{\sqrt{\nin \nout}}
- \left(\sqrt{c} \yin^T\yin+\frac{\yout^T\yout}{\sqrt{c}}\right) \ra 0.
\end{align}
By the PSD property of graph Laplacian matrix, $\yin^T \Lin \yin>0$ and $\yout^T \Lout \yout>0$ if and only if $\yin$ and $\yout$ are not constant vectors. Therefore (\ref{eqn_Lagrange16}) implies $\yin$ and $\yout$ are both constant vectors. Finally, by the constraints $\yin^T\yin+\yout^T\yout=1$ and $\onein^T\yin+\oneout^T\yout=0$, we have when $\pout \leq \pout^*$, almost surely,
\begin{align}
\label{eqn_Lagrange17}
\sqrt{\frac{n \nin}{\nout}} \yin \ra \pm \onein~\text{and}~\sqrt{\frac{n \nout}{\nin}} \yout \ra \mp \oneout.
\end{align}
\end{appendices}

\bibliographystyle{IEEEtran}
\bibliography{IEEEabrv,Deep_community}

\clearpage

\section*{Supplementary File}

\begin{figure}[h]
    \centering
    \includegraphics[width=3.8in]{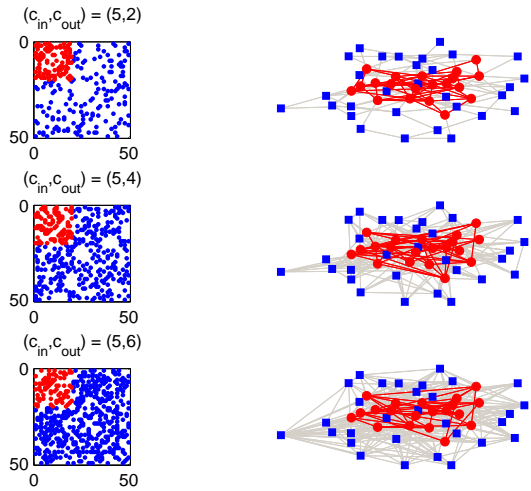}
    \caption{An illustration of deep community detection. The entire network is a realization of the stochastic block model introduced in Sec \ref{Sec_block}, with network size $n=50$ and deep community size $n_{\text{deep}}=20$. The nodes in the deep community are marked by red solid circle, and the other nodes are marked by blue solid rectangles. The left and right columns represent adjacency matrices and their corresponding graphs, respectively. It is observed when $\cin$ is fixed, the deep community is more difficult to be detected as $\cout$ increases.}
    \label{Fig_deep_community_demo}
\end{figure}

\begin{figure}[h]
	\centering
	\includegraphics[width=3.7in]{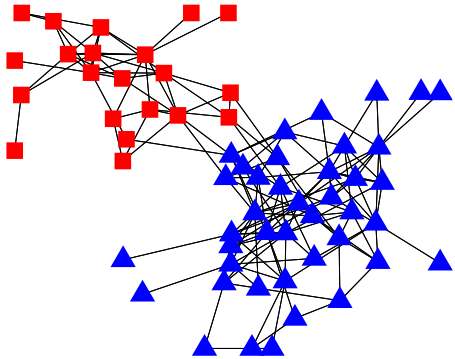}
	\caption{Spectral clustering on dolphin social network. Spectral clustering results in the same discovered communities as the edge-LFVC community detection method.  However, unlike the proposed node-LFVC method it does not explicitly identify the four mixed membership dolphins that connect the two communities.}
	\label{Fig_Dolphin_SC}
\end{figure}

\begin{figure}[h]
	\centering
	\includegraphics[width=3.7in]{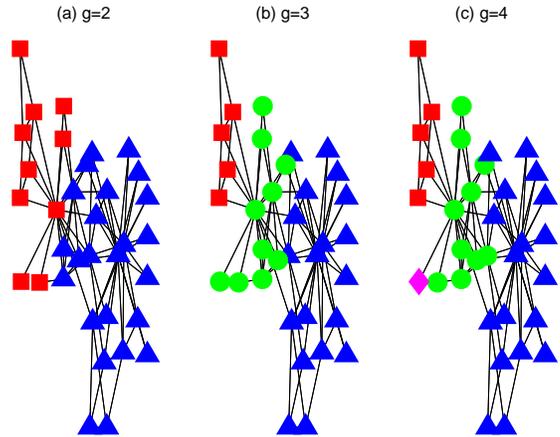}
	\caption{Spectral clustering on Zachary's karate club dataset. The first $g$ smallest eigenvectors of the graph Laplacian matrix are used to cluster the nodes into $g$ communities as suggested in \cite{Luxburg07}. When $g=4$, spectral clustering yields imbalanced communities (one community has single node).}
	\label{Fig_Karate_SC}
\end{figure}

%
%
%

\end{document}